\pdfoutput=1

\documentclass[journal]{IEEEtran}

\usepackage{graphicx}
\usepackage{tikz}
\usetikzlibrary{decorations.pathreplacing, shapes.geometric}
\usepackage{balance}
\usepackage[cmex10]{amsmath}
\usepackage{amsfonts}
\usepackage{xcolor}
\usepackage{algorithm}
\usepackage{algorithmic}
\usepackage{comment}
\usepackage{booktabs}
\usepackage{tabularx}
\usepackage{amsmath, amssymb, amsthm}
\usepackage[caption=false,font=footnotesize]{subfig}

\usepackage{cite}

\newtheorem{theorem}{Theorem}
\newtheorem{lemma}{Lemma}
\newtheorem{corollary}{Corollary}
\newtheorem{remark}{Remark}

\usepackage{microtype}

\begin{document}
\bstctlcite{IEEEexample:BSTcontrol}   

\title{Pinching-Antenna Differential Spatial Modulation with RSSI-Assisted Selection and Low-Complexity Detection: Methods and Performance Analysis} 

\author{Yusuf~Akar,~\IEEEmembership{Student Member,~IEEE,} 
        Mahmoud~Raeisi,~\IEEEmembership{Member,~IEEE,} 
        Henk Wymeersch,~\IEEEmembership{Fellow,~IEEE,}
        Mikko~Valkama,~\IEEEmembership{Fellow,~IEEE,}
        and~Ertugrul~Basar,~\IEEEmembership{Fellow,~IEEE \vspace{-0.5cm}}
\thanks{A preliminary version of this paper was presented at the International Conference on Telecommunications (ICT), Thessaloniki, Greece, May 2026 \cite{11594671}. \looseness=-1}%
\thanks{Y. Akar, M. Valkama, and E. Basar are with the Tampere Wireless Research Centre, Department of Electrical Engineering, Tampere University, 33720 Tampere, Finland (e-mail: yusuf.akar@tuni.fi; mikko.valkama@tuni.fi; ertugrul.basar@tuni.fi).}
\thanks{M. Raeisi and H. Wymeersch are with the Department of Electrical Engineering,
Chalmers University of Technology, Gothenburg, Sweden (e-mail: raeisi@chalmers.se; henkw@chalmers.se).}%

\vspace{-1em}
}

\maketitle

\begin{abstract}
Pinching antenna (PA) systems provide macroscopic spatial reconfigurability by enabling signal radiation from favorable locations along a dielectric waveguide. Exploiting this flexibility, however, typically requires channel knowledge for PA configuration, which can introduce substantial overhead in dynamic propagation environments. This paper proposes PA-assisted differential spatial modulation (PA-DSM), enabling data transfer with spatial index selection and non-coherent data detection without instantaneous channel state information (CSI) at the receiver. To exploit PA reconfigurability while preserving this CSI-efficient operation, we introduce a slow-timescale received signal strength indicator (RSSI)-assisted selection strategy that identifies a favorable subset of PAs from scalar power measurements. We further specialize a two-stage low-complexity maximum-likelihood (LC-ML) detector to phase shift keying (PSK) signaling, obtaining closed-form symbol decisions that remove the joint modulation-combination search while preserving the minimizer of exhaustive differential detection. Furthermore, a moment generating function (MGF)-based pairwise error analysis and the corresponding bit error probability (BEP) bound are developed under normalized common-$K$ Rician fading, revealing how the difference rank and spectrum, number of receive antennas, and previous-state-dependent line-of-sight (LoS) alignment affect the error performance. Numerical results across representative indoor-office, indoor-factory, and street-canyon sixth-generation (6G) Frequency Range 3 (FR3) scenarios validate the analytical trends and quantify when PA-DSM outperforms coherent benchmark schemes whose channel estimates age between pilot transmissions.
\end{abstract}

\begin{IEEEkeywords}
Pinching antenna systems (PASS), differential spatial modulation, RSSI-assisted selection, low-complexity detection.
\end{IEEEkeywords}

\IEEEpeerreviewmaketitle
\vspace{-1em}
\section{Introduction}

\IEEEPARstart{T}{he} forthcoming sixth-generation (6G) wireless networks are envisioned to support high data rates, massive connectivity, and highly reliable low-latency services \cite{11456641}. Meeting these requirements calls for physical layer techniques that efficiently exploit the available spatial, spectral, and propagation-domain degrees of freedom. Multiple-input multiple-output (MIMO) systems improve spectral efficiency and link reliability through spatial multiplexing and diversity \cite{7448967}, while techniques such as non-orthogonal multiple access (NOMA) \cite{10659349} and noise modulation \cite{10373568, 11032161} introduce additional dimensions for resource sharing and information transmission. More recently, attention has shifted toward architectures that make the effective wireless channel more controllable through programmable propagation or antenna-position adaptation \cite{8796365, 10318061, 9264694}.

Within this paradigm, pinching antenna (PA) systems provide macroscopic spatial reconfigurability by enabling signal radiation from selected locations along a dielectric waveguide \cite{NTTDOCOMO}. This capability is attractive for Frequency Range 3 (FR3) deployments, which combine sub-6-GHz-like coverage with millimeter-wave capacity potential but remain affected by propagation loss, blockage, and channel variation \cite{10605910, 11327450}. These characteristics make PA systems a promising platform for advanced transmission techniques that jointly exploit signal and spatial domains.

\vspace{-0.5em}
\subsection{Related Works}

Existing research on PA systems (PASS) can be broadly grouped into foundational channel modeling and architecture studies and application-oriented PA-assisted communication designs. At the architectural level, recent studies have characterized the potential spectral- and energy-efficiency tradeoffs of PAs relative to reconfigurable intelligent surfaces (RIS)-based deployments in certain mmWave scenarios \cite{11172334}. Channel modeling and transceiver design have also been investigated by jointly accounting for deterministic in-waveguide propagation and stochastic wireless fading \cite{11368709}, providing a physical layer foundation for subsequent PA-assisted communication schemes.

Building on these foundations, recent studies have integrated PAs into optimization-based system designs, multiuser transmission, and MIMO architectures. Representative examples include RIS-assisted multi-waveguide PASS with joint PA-position and RIS-phase optimization \cite{he2025risassisteddownlinkpinchingantennasystems}, PA-assisted partial NOMA \cite{11204499}, and MIMO-PASS architectures with iterative hybrid beamforming \cite{11414143}. Other application-oriented directions include PA-assisted integrated sensing and communication (ISAC) and physical-layer security, with designs addressing dual-functional radar communication \cite{11303890}, secure directional and index modulation (IM) \cite{11205176}, and IM-enabled localization and data transmission \cite{11314615}.

Among these directions, PA-assisted IM is most closely related to the present work because PA reconfigurability naturally creates distinguishable spatial index states. IM conveys information through the indices of communication resources in addition to conventional modulation symbols, while spatial modulation (SM) specifically maps information to the active transmit antenna indices \cite{8004416}. SM is well suited to series-fed PASS because single-PA activation avoids the non-uniform power distribution associated with simultaneous radiation and simplifies the RF implementation. 
Existing PA-assisted IM research has formulated composite waveguide--wireless channel models and developed low-complexity (LC) sphere decoding methods for coherent detection \cite{11368709}.

Most existing PA-assisted IM/SM frameworks employ coherent detection and therefore require instantaneous channel state information (CSI) at the receiver to recover the active spatial index and transmitted symbols. Maintaining sufficiently accurate CSI may require frequent pilot transmission in deployments affected by mobility, blockage, or channel variation. This consideration is particularly relevant in FR3 environments, where measurement studies report non-negligible multipath and non-line-of-sight (NLoS) components across indoor, industrial, and outdoor scenarios \cite{10605910,11160744,11161884}. Differential SM (DSM) provides an alternative by encoding information in the transitions between consecutive transmission blocks, thereby enabling non-coherent data detection without explicit instantaneous CSI acquisition \cite{8004416}.  

However, exploiting DSM in a single-waveguide PA system introduces two coupled challenges: selecting a favorable PA subset without estimating the complete complex channel matrix and managing the joint search over spatial permutations and modulation symbols. Unlike the fluid antenna-assisted differential IM scheme in \cite{11184847} and the concurrent Alamouti-based dual-waveguide DSM design in \cite{Tao_DSM_PA}, this work considers single-waveguide, single-RF-chain PA-DSM with RSSI-assisted PA selection and ML-equivalent LC detection.
\vspace{-0.3cm}
\subsection{Motivation and Contributions}

Addressing the identified research gap requires resolving two coupled design challenges. First, the macroscopic reconfigurability of a PA waveguide must be exploited without estimating the complete instantaneous complex channel between all candidate PAs and the receiver. Second, the conventional DSM receiver jointly searches over $Q$ spatial permutations and $M^{N_t}$ modulation-symbol combinations, where $M$ is the modulation order and $N_t$ is the number of selected PAs used for data transfer with spatial index selection. The resulting $QM^{N_t}$ hypotheses limit the scalability of exhaustive detection.

Motivated by these challenges, we propose a novel PA-DSM system that exploits the macroscopic spatial reconfigurability of PAs while retaining the non-coherent detection capability of differential signaling. An uplink beacon provides scalar RSSI measurements for selecting a favorable PA subset, which is reused across multiple PA-DSM blocks and updated only over the large-scale coherence timescale. This avoids per-block selection overhead and instantaneous complex CSI acquisition. At the receiver, a pre-calculated-symbol-metric LC-ML detector separates the modulation-symbol and spatial-permutation searches without omitting any candidate hypotheses and therefore yields the same decision as exhaustive detection. The resulting system supports long-term macroscopic PA adaptation and non-coherent data detection with reduced receiver complexity. \looseness=-1

The main contributions of this paper are summarized as follows:
\begin{itemize}
\item \textit{Single-Waveguide PA-DSM Architecture}: We develop a single-user PA-DSM model that incorporates in-waveguide propagation, composite wireless fading, spatial-permutation signaling, and differential encoding. The transmission structure activates one PA per time slot and enables non-coherent data detection without instantaneous receiver CSI. 

\item \textit{RSSI-Assisted PA Selection}: Scalar RSSI measurements from an uplink probing beacon select the $N_t$ favorable candidate PAs over the large-scale coherence timescale. The subset is reused across multiple PA-DSM blocks and updated only when the large-scale propagation conditions change appreciably, avoiding per-block selection and instantaneous complex CSI estimation.

\item \textit{PSK-Specialized Low-Complexity ML Detection}: Building on \cite{10855589}, we derive closed-form PSK symbol decisions for the column-pair metrics. The resulting detector replaces the $M^{N_t}$ joint modulation search with $N_t^2$ phase quantizations and evaluation of the $Q$ admissible permutations, while preserving the minimizer of the exhaustive differential metric.

\item \textit{Analytical and Numerical Characterization}: We derive a moment generating function (MGF)-based pairwise error characterization and a bit error probability (BEP) union bound under normalized common-$K$ Rician fading. The analysis separates the scattered and deterministic LoS contributions, showing that the previous differential state preserves the pairwise rank and eigenvalues but can change the LoS contribution through spatial alignment. It further identifies the rank-dependent polynomial signal-to-noise ratio (SNR) exponent, codebook minimum-rank behavior, and pairwise coding coefficient. A complexity analysis further quantifies the resulting saving, which removes the exponential $M^{N_t}$ scaling of exhaustive detection.

\end{itemize}

\textit{Notation:} Bold lowercase and uppercase letters denote vectors and matrices, and $[\mathbf{A}]_{a,b}$ is the $(a,b)$th entry of $\mathbf{A}$. The operators $(\cdot)^{\mathsf T}$, $(\cdot)^{H}$, $(\cdot)^{*}$, $\operatorname{Tr}(\cdot)$, $\Re\{\cdot\}$, $\mathbb{E}[\cdot]$, $\operatorname{Var}[\cdot]$, $\lfloor\cdot\rfloor$, and $\mathcal{O}(\cdot)$ denote transpose, Hermitian transpose, conjugate, trace, real part, expectation, variance, floor, and asymptotic order, respectively, while $\|\cdot\|$, $\|\cdot\|_F$, and $\odot$ denote the Euclidean norm, Frobenius norm, and Hadamard product. Here, $\operatorname{diag}(\mathbf{a})$ is a diagonal matrix formed from $\mathbf{a}$, $\mathbf{I}_N$ and $\mathbf{1}_{M\times N}$ are the identity and all-ones matrices, $\mathbb{C}^{M\times N}$ and $\mathbb{R}^{M\times N}$ are the sets of complex and real $M\times N$ matrices, and $\mathcal{N}(\mu,\sigma^2)$ and $\mathcal{CN}(\mu,\sigma^2)$ are real and circularly symmetric complex Gaussian distributions. Finally, $\mathcal{S}$ and $\mathcal{P}$ denote the PSK alphabet and permutation codebook, and $P_{\mathrm a}(\cdot)$ denotes the pairwise error probability (PEP).

\vspace{-0.3cm}
\section{System and Channel Model}
\label{sec:system_model}
\vspace{-0.1cm}
This section presents the proposed PA-DSM system architecture, characterizes the underlying propagation channel, and establishes the corresponding signal model.
\vspace{-0.3cm}
\subsection{System Model}
\vspace{-0.5em}
The considered downlink communication architecture features a base station (BS) equipped with a single radio frequency (RF) chain feeding a longitudinal dielectric waveguide elevated at $z_{wg}$. To enable macroscopically flexible and discrete spatial signal radiation, the waveguide is densely populated with $N_{\text{all}}$ electronically switchable PA tap points distributed at uniform intervals \cite{11202577}.\footnote{In practice, electronically controlled tap points may be implemented using active switching or actuating elements that modify the local boundary conditions, as considered in leaky-coaxial-cable architectures \cite{11657464}.} 

\begin{figure}[!t]
    \centering
    \includegraphics[scale = 0.48]{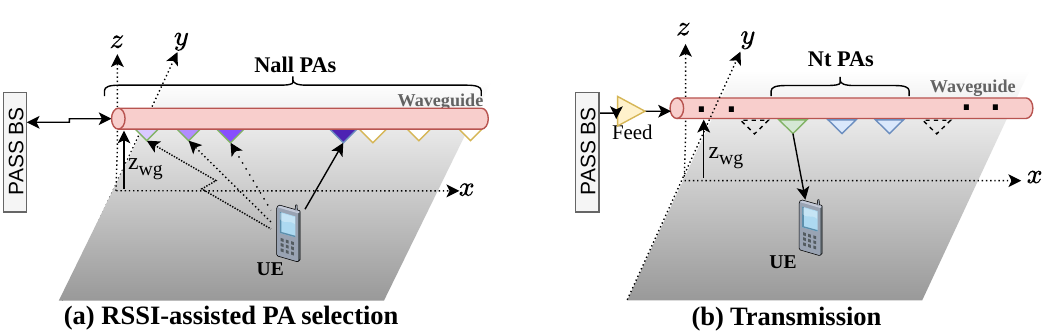}
    \caption{Two-stage operation of the proposed PA-DSM architecture. (a) The UE transmits an uplink probing beacon and the BS selects the $N_t$ PAs with the strongest measured RSSI values from $N_{\rm all}$ candidate locations. (b) During PA-DSM transmission, the selected PAs are activated sequentially according to the spatial-permutation index, with only one PA radiating per time slot.}
    \label{fig:system_model}
    
\end{figure}

The proposed architecture selects $N_t$ PAs from $N_{\rm all}$ candidate locations. Before PA-DSM transmission, a designated UE antenna repeatedly transmits an unmodulated beacon in the TDD uplink, while the BS sequentially connects its RF chain to each candidate PA and measures the received power through the corresponding PA--waveguide path. Following \cite{Zhang2025PASS_RSSI}, candidate $n$ is characterized by $P_{r,n}\triangleq L_p^{-1}\sum_{\ell=1}^{L_p}|r_n^{p}[\ell]|^2$, where $r_n^{p}[\ell]$ is the $\ell$th of $L_p$ received beacon samples averaged to suppress small-scale fading, and $L_p=4$ throughout. The active set $\mathcal A_{\rm active}$ contains the indices of the $N_t$ largest RSSI values. This strength-oriented rule avoids instantaneous complex CSI but does not explicitly optimize the pairwise separation among PA-DSM hypotheses. A complete sweep requires $N_{\mathrm{all}}L_p$ beacon samples, and the selected subset is reused across multiple PA-DSM blocks to amortize this probing cost.

Whenever the selected subset changes, differential encoding is restarted from $\mathbf{S}_0=\mathbf{I}_{N_t}$, so each update costs one additional reference block. The relative cost of this probing follows from the separation between the two adaptation timescales. A coherent receiver expends $N_t$ pilot symbols per pilot interval $T_p$, whereas the proposed scheme expends $N_{\mathrm{all}}L_p$ probing samples and one $N_t$-slot reference block per large-scale decorrelation interval $T_{\mathrm{ls}}=d_{\mathrm{dec}}/v$, where $d_{\mathrm{dec}}$ denotes an assumed shadowing decorrelation distance, and $v$ is the UE speed. The selection overhead relative to coherent pilot overhead is therefore $\big[(N_{\mathrm{all}}L_p+N_t)/N_t\big (T_p/T_{\mathrm{ls}})$, which for the parameters of Table~\ref{tab:sim_params} evaluates to $0.51\%$ in the street-canyon case ($N_{\mathrm{all}}=100$) and $0.085\%$ indoors ($N_{\mathrm{all}}=16$).\looseness=-1

Following PA selection, the $N_t$ selected PAs are used for data transfer with spatial index selection, with only one PA radiating at each time instant. Their activation order across each PA-DSM transmission block is determined by the spatial index bits, as detailed in Section~\ref{subsec:signal_model}. The UE is equipped with $N_r$ receive antennas and located at $\mathbf{u}=[u_x,u_y,u_z]^{\mathsf{T}}$ relative to the waveguide feed point at $\mathbf{u}_f=[0,D_y/2,z_{\rm wg}]^{\mathsf{T}}$. The resulting PA-DSM architecture is illustrated in Fig.~\ref{fig:system_model}.

\vspace{-0.1cm}
\subsection{Channel Model}
\label{subsec:channel_model}

The PA-assisted channel incorporates deterministic in-waveguide propagation and Rician wireless fading. Let $\mathbf{H} \in \mathbb{C}^{N_r \times N_t}$ denote the wireless channel matrix between the $N_t$ active PAs and the $N_r$ receive antennas. The channel coefficient between the $j$th active PA and the $i$th receive antenna is modeled as \cite{11368709}
\setlength{\abovedisplayskip}{4pt}
\setlength{\belowdisplayskip}{4pt}
\begin{equation}
    h_{i,j} = \sqrt{\beta_{i,j}} \left( \sqrt{\frac{K_{i,j}}{K_{i,j}+1}} h^{\mathrm{LoS}}_{i,j} + \sqrt{\frac{1}{K_{i,j}+1}} h^{\mathrm{NLoS}}_{i,j} \right) \gamma_{j}.
\end{equation}
Here, $\beta_{i,j}$ denotes the large-scale channel gain incorporating distance-dependent path loss and shadowing, while $K_{i,j}$ is the Rician $K$-factor representing the power ratio between the LoS and NLoS components. For a blocked LoS link, $K_{i,j}=0$, reducing the wireless channel to the NLoS case.

Following~\cite{11368709}, the deterministic LoS coefficient is the free-space phasor $h^{\mathrm{LoS}}_{i,j}=e^{-j2\pi d_{\mathrm{free},i,j}/\lambda}$, whose dependence on the PA-to-antenna distance $d_{\mathrm{free},i,j}$ carries the receive-array response, while $h^{\mathrm{NLoS}}_{i,j}\sim\mathcal{CN}(0,1)$ is the normalized scattered component. The unit-magnitude normalization $|h^{\mathrm{LoS}}_{i,j}|=1$ is retained in the analysis of Section~IV. In addition to the wireless propagation, the signal experiences a deterministic phase shift along the dielectric waveguide, modeled as $\gamma_j = e^{-jk_zd_{{\rm wg},j}}$, where $d_{{\rm wg},j}$ denotes the in-waveguide distance from the RF feed point to the $j$th PA. The longitudinal propagation constant is $k_z=2\pi n_{\rm eff}/\lambda$, where $n_{\rm eff}>1$ is the effective refractive index of the dielectric \cite{11368709}, and $\lambda$ is the carrier wavelength.

Collecting the individual channel coefficients, the channel matrix can be compactly expressed as \cite{11368709}
\begin{equation}
    \mathbf{H}
    =
    \mathbf{B}
    \odot
    \left(
    \mathbf{K}_{\mathrm{LoS}} \odot \mathbf{H}^{\mathrm{LoS}}
    +
    \mathbf{K}_{\mathrm{NLoS}} \odot \mathbf{H}^{\mathrm{NLoS}}
    \right)
    \odot
    \mathbf{\Gamma},
\end{equation}
where $\mathbf{B}$, $\mathbf{\Gamma}$, $\mathbf{H}^{m}$, and $\mathbf{K}_{m}$, for $m\in\{\mathrm{LoS},\mathrm{NLoS}\}$, denote the large-scale fading, in-waveguide phase-shift, wireless-channel, and Rician scaling matrices, respectively. Their $(i,j)$th entries are
$[\mathbf{B}]_{i,j}=\sqrt{\beta_{i,j}}$,
$[\mathbf{\Gamma}]_{i,j}=\gamma_j$,
$[\mathbf{H}^{m}]_{i,j}=h^{m}_{i,j}$,
$[\mathbf{K}_{\mathrm{LoS}}]_{i,j}=\sqrt{K_{i,j}/(K_{i,j}+1)}$, and $[\mathbf{K}_{\mathrm{NLoS}}]_{i,j}=\sqrt{1/(K_{i,j}+1)}$.

\vspace{-0.2cm}
\subsection{Signal Model}
\label{subsec:signal_model}

PA-DSM conveys information through both the spatial and signal domains while enabling non-coherent detection through differential encoding. The spatial-index component encodes information in the PA activation patterns, thereby enhancing spectral efficiency, while differential encoding eliminates the need for instantaneous CSI acquisition at the receiver.

\subsubsection{Signal Construction and Spectral Efficiency}

Transmission is organized into sequential blocks indexed by $k$, each comprising $T=N_t$ time slots. Following the differential modulation structure in \cite{6879496}, this block length yields a square unitary information matrix $\mathbf{X}_k \in \mathbb{C}^{N_t \times N_t}$, constructed as
\setlength{\abovedisplayskip}{4pt}
\setlength{\belowdisplayskip}{4pt}
\begin{equation}
    \mathbf{X}_k = \mathbf{A}_{q_k}\mathbf{D}_k,
    \label{eq:information_matrix}
\end{equation}
where the permutation matrix $\mathbf{A}_{q_k}$ specifies the PA activation sequence and the diagonal matrix $\mathbf{D}_k$ contains the $M$-PSK symbols. Accordingly, the information bits in each block are partitioned into spatial-index bits $B_I$ and modulation-symbol bits $B_S$.

\textbf{Spatial Bits ($B_I$):} A group of $\lfloor \log_2(N_t!) \rfloor$ spatial-index bits is mapped to the permutation index $q_k$, which selects a matrix $\mathbf{A}_{q_k}$ from the predefined codebook $\mathcal{P} = \{\mathbf{P}_1,\dots,\mathbf{P}_Q\}$ of size $Q = 2^{\lfloor \log_2(N_t!) \rfloor}$. The selected permutation matrix determines the activation sequence of the $N_t$ PAs across the $T$ time slots, where $[\mathbf{A}_{q_k}]_{j,t}=1$ indicates that the $j$th active PA radiates during the $t$th time slot of block $k$. Table~\ref{tab:lookup_nt3} illustrates this mapping for $N_t=3$, where two spatial-index bits select one of $Q=4$ activation patterns from the $N_t!=6$ possible permutations. For $N_t=4$, the $Q=16$ patterns are instead selected greedily from the $N_t!=24$ permutations by starting from the identity and repeatedly adding the permutation whose minimum Hamming distance to the already selected patterns is largest, with ties resolved lexicographically. In both cases the spatial-index bits are Gray-mapped to the listed order.

\textbf{Modulation Bits ($B_S$):} The $N_t\log_2(M)$ modulation bits are mapped to a symbol vector $\mathbf{s}_k=[s_{k,1},\dots,s_{k,N_t}]^{\mathsf{T}}$, whose entries are drawn from an $M$-ary PSK alphabet. These symbols form the diagonal matrix $\mathbf{D}_k=\operatorname{diag}(\mathbf{s}_k)$. Consequently, the information matrix $\mathbf{X}_k=\mathbf{A}_{q_k}\mathbf{D}_k$ contains exactly one nonzero entry in each row and column. Thus, only one PA radiates in each time slot, while each of the $N_t$ selected PAs is activated exactly once per transmission block. This single-PA activation avoids the non-uniform radiated-power distribution associated with simultaneous radiation from multiple PAs along a series-fed dielectric waveguide \cite{11202577}. Since each block conveys $B_I+B_S=\lfloor\log_2(N_t!)\rfloor+N_t\log_2(M)$ bits over $N_t$ time slots, the resulting spectral efficiency is
\setlength{\abovedisplayskip}{4pt}
\setlength{\belowdisplayskip}{4pt}
\begin{equation}
    \eta
    =
    \frac{\lfloor\log_2(N_t!)\rfloor}{N_t}
    +
    \log_2(M)
    \quad \text{bps/Hz}.
    \label{eq:spectral_eff}
\end{equation}

\begin{table}[t] 
\centering 
\footnotesize 
\caption{Permutation mapping for $N_t=3$, adapted from \cite{6879496}.} \label{tab:lookup_nt3} 
\setlength{\tabcolsep}{3pt} 
\renewcommand{\arraystretch}{1.2} 
\begin{tabular}{|c|c|c|c|c|} 
\hline \textbf{Bits} & \textbf{00} & \textbf{01} & \textbf{11} & \textbf{10} \\ 
\hline
$\mathbf{A}_{q_k}$ &
$\begin{bmatrix} 1&0&0\\ 0&1&0\\ 0&0&1 \end{bmatrix}$ & 
$\begin{bmatrix} 1&0&0\\ 0&0&1\\ 0&1&0 \end{bmatrix}$ & 
$\begin{bmatrix} 0&0&1\\ 1&0&0\\ 0&1&0 \end{bmatrix}$ & 
$\begin{bmatrix} 0&1&0\\ 1&0&0\\ 0&0&1 \end{bmatrix}$ \\ 
\hline 
\end{tabular} 
 
\end{table}

\subsubsection{Differential Encoding and Detection}

To enable non-coherent detection, the information matrix $\mathbf{X}_k$ is differentially encoded across consecutive transmission blocks. Specifically, the transmitted signal matrix $\mathbf{S}_k \in \mathbb{C}^{N_t \times N_t}$ evolves recursively as \looseness=-1
\setlength{\abovedisplayskip}{6pt}
\setlength{\belowdisplayskip}{6pt}
\begin{equation}
    \mathbf{S}_k = \mathbf{S}_{k-1}\mathbf{X}_k,
    \label{eq:differential_encoding}
\end{equation}
where $\mathbf{S}_0=\mathbf{I}_{N_t}$ is the known initialization matrix. Since $\mathbf{S}_0$ and $\mathbf{X}_k$ are monomial unitary matrices, $\mathbf{S}_k$ preserves this structure for all $k$. Hence, $\mathbf{S}_k$ contains exactly one nonzero entry in each row and column, ensuring that only one PA radiates in each time slot.

The received signals over the $N_t$ time slots of block $k$ are collected column-wise in $\mathbf{Y}_k \in \mathbb{C}^{N_r \times N_t}$ as 
\setlength{\abovedisplayskip}{4pt}
\setlength{\belowdisplayskip}{4pt}
\begin{equation}
    \mathbf{Y}_k
    =
    \sqrt{P_t}\mathbf{H}\mathbf{S}_k
    +
    \mathbf{N}_k,
    \label{eq:received_signal}
\end{equation}
where $P_t$ denotes the average transmit power and $\mathbf{N}_k \in \mathbb{C}^{N_r \times N_t}$ is the additive white Gaussian noise (AWGN) matrix with independent entries $[\mathbf{N}_k]_{i,t} \sim \mathcal{CN}(0,\sigma_n^2)$, where $\sigma_n^2$ denotes the noise variance.
Assuming that the channel remains quasi-static over two consecutive transmission blocks, substituting \eqref{eq:differential_encoding} into \eqref{eq:received_signal} yields the differential input--output relationship \looseness=-1
\setlength{\abovedisplayskip}{4pt}
\setlength{\belowdisplayskip}{4pt}
\begin{equation}
    \mathbf{Y}_k
    =
    \mathbf{Y}_{k-1}\mathbf{X}_k
    +
    \widetilde{\mathbf{N}}_k, 
    \label{eq:differential_received}
\end{equation}
where 
\setlength{\abovedisplayskip}{4pt}
\setlength{\belowdisplayskip}{4pt}
\begin{equation}
    \widetilde{\mathbf{N}}_k
    =
    \mathbf{N}_k-\mathbf{N}_{k-1}\mathbf{X}_k 
    \label{eq:equivalent_noise}
\end{equation}
denotes the equivalent noise induced by differential processing. Based on \eqref{eq:differential_received}, the conventional differential detector estimates the information matrix by minimizing the squared Euclidean distance over the complete PA-DSM codebook $\mathcal{C}$ as
\setlength{\abovedisplayskip}{4pt}
\setlength{\belowdisplayskip}{4pt}
\begin{equation}
    \hat{\mathbf{X}}_k
    =
    \arg\min_{\mathbf{X}\in\mathcal{C}}
    \left\|
        \mathbf{Y}_k-\mathbf{Y}_{k-1}\mathbf{X}
    \right\|_F^2. 
    \label{eq:ml_detector}
\end{equation}
The metric in \eqref{eq:ml_detector} is formed from the received blocks $\mathbf{Y}_k$ and $\mathbf{Y}_{k-1}$ rather than from previously detected information matrices, so the scheme exhibits no decision-feedback error propagation. Adjacent blocks nevertheless share the noise realization contained in $\mathbf{Y}_{k-1}$, and this dependence is treated in the pairwise error analysis of Section~\ref{subsec:cpep}. Because each candidate information matrix is formed as $\mathbf{X}=\mathbf{A}_q\mathbf{D}$, the detector jointly searches over $Q$ spatial permutations and $M^{N_t}$ modulation-symbol combinations, resulting in $Q M^{N_t}$ candidate hypotheses. The resulting complexity grows rapidly with $N_t$ and $M$, motivating the low-complexity ML-equivalent detection strategy developed in Section~\ref{sec:detector_complexity}. 
\vspace{-0.3em}
\section{Low-Complexity Maximum-Likelihood Decoding for PA-DSM}
\label{sec:detector_complexity}
\vspace{-0.2em}
The exhaustive detector in Section~\ref{subsec:signal_model} jointly searches over $Q$ spatial permutations and $M^{N_t}$ modulation-symbol combinations. We employ the two-stage LC-ML metric decomposition in \cite{10855589} and specialize it to PSK signaling by exploiting the separability of the differential metric and the constant-modulus constellation. The resulting detector obtains the exact minimizing symbols through PSK phase quantization and then evaluates the $Q$ admissible permutations. Consequently, it returns the same minimizer of \eqref{eq:ml_detector} as exhaustive joint search while eliminating the $M^{N_t}$ modulation-combination search.

The detector exploits the structure $\mathbf X_k=\mathbf A_{q_k}\mathbf D_k$. In \eqref{eq:differential_received}, $\mathbf A_{q_k}$ reorders the columns of $\mathbf Y_{k-1}$, while $\mathbf D_k$ scales the reordered columns by the PSK symbols. Hence, the optimal symbol for each column pairing can be determined independently before evaluating the admissible spatial permutations, leading to the two-stage procedure in Algorithm~\ref{alg:lc_ml}.

\begin{algorithm}[t] 
\footnotesize 
\caption{LC-ML Detector for PA-DSM} 
\label{alg:lc_ml} 
\begin{algorithmic}[1] 
\STATE \textbf{Input:} $\mathbf{Y}_k$, $\mathbf{Y}_{k-1}$, permutation codebook $\mathcal{P}$, and PSK alphabet $\mathcal{S}$ 
\STATE \textbf{Phase 1: PSK Metric Pre-Computation} 
\FOR{$t=1$ to $N_t$} 
    \FOR{$l=1$ to $N_t$} 
    \STATE $z_{t,l} \leftarrow \mathbf y_{k-1,l}^{H}\mathbf y_{k,t}$ 
    \STATE $\hat{s}_{t,l} \leftarrow \mathcal Q_{\mathcal S}(z_{t,l})$ 
    \STATE $[\mathbf C]_{t,l} \leftarrow \|\mathbf y_{k,t}\|_2^2+ \|\mathbf y_{k-1,l}\|_2^2 -2\Re\{\hat{s}_{t,l}^*z_{t,l}\}$ 
    \ENDFOR 
\ENDFOR
\STATE \textbf{Phase 2: Spatial-Permutation Search} 
\FOR{$q=1$ to $Q$} 
    \STATE $\Lambda_q\leftarrow0$ 
    \FOR{$t=1$ to $N_t$} 
        \STATE $l_{q,t} \leftarrow \text{the unique }j \text{ such that } [\mathbf{A}_q]_{j,t}=1$ 
        \STATE $\Lambda_q \leftarrow \Lambda_q+[\mathbf{C}]_{t,l_{q,t}}$ 
    \ENDFOR 
\ENDFOR 
\STATE $\hat{q} \leftarrow \displaystyle \arg\min_{q\in\{1,\ldots,Q\}} \Lambda_q$ 
\STATE $\hat{\mathbf{s}}_k \leftarrow [\hat{s}_{1,l_{\hat q,1}},\ldots, \hat{s}_{N_t,l_{\hat q,N_t}}]^{\mathsf T}$ 
\STATE \textbf{Output:} $\hat{\mathbf{X}}_k = \mathbf{A}_{\hat q} \operatorname{diag}(\hat{\mathbf{s}}_k)$
\end{algorithmic} 
 
\end{algorithm}

\vspace{-0.5em}
\subsection{Phase 1: Symbol Metric Pre-Computation}

For each of the $N_t^2$ column pairings $(t,l)$, define $z_{t,l}\triangleq \mathbf y_{k-1,l}^{H}\mathbf y_{k,t}$. Since $|s|=1$ for $s\in\mathcal S$,

\begin{equation} 
\begin{aligned}
    d_{t,l}(s) &\triangleq \left\| \mathbf y_{k,t}-\mathbf y_{k-1,l}s \right\|_2^2\\ &= \|\mathbf y_{k,t}\|_2^2+ \|\mathbf y_{k-1,l}\|_2^2 -2\Re\{s^*z_{t,l}\}. 
\end{aligned} 
\label{eq:expanded_symbol_metric} 
\end{equation} 
Thus, 
\begin{equation} 
\begin{aligned} 
    \hat{s}_{t,l} &= \mathcal Q_{\mathcal S}(z_{t,l}), &[\mathbf C]_{t,l} = d_{t,l}(\hat{s}_{t,l}), 
\end{aligned} 
\label{eq:closed_form_symbol_metric} 
\end{equation}
where $\mathcal Q_{\mathcal S}(\cdot)$ returns the PSK symbol nearest in phase. These costs and symbols are reused across all admissible permutation hypotheses.

\subsection{Phase 2: Spatial Hypothesis Evaluation and Decision}

For each admissible permutation $\mathbf A_q\in\mathcal P$, let $l_{q,t}$ denote the unique column index satisfying $[\mathbf A_q]_{l_{q,t},t}=1$. Using the costs obtained in Phase~1, the permutation metric and decision are
\begin{equation}
\Lambda_q = \sum_{t=1}^{N_t} [\mathbf C]_{t,l_{q,t}}, \qquad \hat q = \arg\min_{q\in\{1,\ldots,Q\}} \Lambda_q. 
\label{eq:permutation_metric_decision} 
\end{equation} 
The corresponding symbols are recovered as $\hat{s}_{k,t}=\hat{s}_{t,l_{\hat q,t}}$, $t=1,\ldots,N_t$. Since Phase~1 obtains the exact minimizing PSK symbol for every column pairing and Phase~2 evaluates all $Q$ admissible permutations, the two-stage detector returns the same minimizer of \eqref{eq:ml_detector} as exhaustive joint search.

\vspace{-0.5em}

\section{Theoretical Performance Analysis}
\label{sec:theoretical_analysis}
 
This section develops a tractable analytical characterization of PA-DSM under a normalized common-$K$ Rician model. We first formulate the pairwise error probability (PEP) of the differential detector. Since the equivalent noise and the previous received block share the noise realization $\mathbf{N}_{k-1}$, their statistical dependence prevents a direct exact characterization. Therefore, a high-SNR approximation is adopted to obtain a tractable conditional PEP expression. The conditional PEP is then averaged over the normalized Rician channel using Craig's representation of the Gaussian $Q$-function \cite{258319}. The resulting analysis reveals the impact of the pairwise difference structure, receive diversity, and previous differential state on the error performance. \looseness=-1
\vspace{-0.5em}
\subsection{Conditional Pairwise Error Probability}
\label{subsec:cpep}

This subsection derives the conditional PEP of the differential detector, which forms the basis for the subsequent channel averaging and asymptotic analysis. Because the detector metric depends on the equivalent noise in \eqref{eq:equivalent_noise}, we first characterize its statistical properties. Since $\mathbf{X}_k$ is unitary, multiplication of $\mathbf{N}_{k-1}$ by $\mathbf{X}_k$ preserves its covariance. Moreover, since $\mathbf{N}_k$ and $\mathbf{N}_{k-1}$ are independent, each entry of the equivalent noise matrix $\widetilde{\mathbf N}_k$ defined in \eqref{eq:equivalent_noise} has marginal variance
\setlength{\abovedisplayskip}{4pt}
\setlength{\belowdisplayskip}{4pt}
\begin{equation}
    \sigma_{\widetilde{n}}^2 = 2\sigma_n^2.
    \label{eq:equivalent_noise_variance}
\end{equation}
For the theoretical analysis, the large-scale channel gain is normalized such that the resulting channel coefficients satisfy $\mathbb{E}[|h_{i,j}|^2]=1$. Accordingly, the nominal SNR is defined as
\setlength{\abovedisplayskip}{4pt}
\setlength{\belowdisplayskip}{4pt}
\begin{equation}
    \rho \triangleq \frac{P_t}{\sigma_n^2}.
    \label{eq:snr_definition}
\end{equation}
The effective SNR associated with the differential noise is therefore $\rho_{\rm eff}=\rho/2$.

Consider the pairwise error event in which the detector selects $\hat{\mathbf{X}}$ instead of the transmitted matrix $\mathbf{X}$. Conditioned on the channel realization $\mathbf{H}$ and the previous differential state $\mathbf{S}_{k-1}$, the pairwise error event associated with the decision rule in \eqref{eq:ml_detector} is
\setlength{\abovedisplayskip}{4pt}
\setlength{\belowdisplayskip}{4pt}
\begin{equation} 
    \lVert \mathbf Y_k-\mathbf Y_{k-1}\hat{\mathbf X} \rVert_F^2 < \lVert \mathbf Y_k-\mathbf Y_{k-1}\mathbf X \rVert_F^2. 
\label{eq:pairwise_error_event} 
\end{equation}
Substituting the differential input--output relation \eqref{eq:differential_received} into \eqref{eq:pairwise_error_event} and defining the difference matrix $\mathbf{\Delta} \triangleq \mathbf{X}-\hat{\mathbf{X}}$, the pairwise error event becomes
\setlength{\abovedisplayskip}{4pt}
\setlength{\belowdisplayskip}{4pt}
\begin{equation}
    \left\|
        \mathbf{Y}_{k-1}\mathbf{\Delta}
        +
        \widetilde{\mathbf{N}}_k
    \right\|_F^2
    <
    \left\|
        \widetilde{\mathbf{N}}_k
    \right\|_F^2.
    \label{eq:pairwise_error_expanded}
\end{equation}
Expanding the squared norm and canceling the common noise-energy term yields the exact error event
\setlength{\abovedisplayskip}{4pt}
\setlength{\belowdisplayskip}{4pt}
\begin{equation}
    Z
    =
    \left\|
        \mathbf{Y}_{k-1}\mathbf{\Delta}
    \right\|_F^2
    +
    2\Re
    \left\{
        \operatorname{Tr}
        \left(
            \widetilde{\mathbf{N}}_k^{H}
            \mathbf{Y}_{k-1}\mathbf{\Delta}
        \right)
    \right\}
    <0.
    \label{eq:exact_error_event}
\end{equation}
The matrices $\widetilde{\mathbf{N}}_k$ and $\mathbf{Y}_{k-1}$ are statistically correlated because both contain the noise matrix $\mathbf{N}_{k-1}$. Under the high-SNR approximation, the contribution of the noise in the previous received block is neglected,
\setlength{\abovedisplayskip}{4pt}
\setlength{\belowdisplayskip}{4pt}
\begin{equation}
    \mathbf{Y}_{k-1}
    \approx
    \sqrt{P_t}\mathbf{H}\mathbf{S}_{k-1}.
    \label{eq:high_snr_approximation}
\end{equation}

Substituting \eqref{eq:high_snr_approximation} into \eqref{eq:exact_error_event} gives the approximate error event $Z_a<0$, where
\begin{equation}
    Z_a
    =
    P_t
    \left\|
        \mathbf{H}\mathbf{S}_{k-1}\mathbf{\Delta}
    \right\|_F^2
    +
    2\sqrt{P_t}\Re
    \left\{
        \operatorname{Tr}
        \left(
            \widetilde{\mathbf{N}}_k^{H}
            \mathbf{H}\mathbf{S}_{k-1}\mathbf{\Delta}
        \right)
    \right\}.
\label{eq:approximate_error_event}
\end{equation}
Conditioned on $\mathbf{H}$ and $\mathbf{S}_{k-1}$, and using \eqref{eq:equivalent_noise_variance}, the conditional mean and variance of $Z_a$ are
\setlength{\abovedisplayskip}{4pt}
\setlength{\belowdisplayskip}{4pt}
\begin{equation}
    \mu_Z = 
    \mathbb{E}
    \left[
        Z_a \mid \mathbf{H},\mathbf{S}_{k-1}
    \right]
    =
    P_t
    \left\|
        \mathbf{H}\mathbf{S}_{k-1}\mathbf{\Delta}
    \right\|_F^2,
    \label{eq:za_mean}
\end{equation}
and
\begin{equation}
    \sigma_Z^2 = 
    \operatorname{Var}
    \left[
        Z_a \mid \mathbf{H},\mathbf{S}_{k-1}
    \right]
    =
    4P_t\sigma_n^2
    \left\|
        \mathbf{H}\mathbf{S}_{k-1}\mathbf{\Delta}
    \right\|_F^2.
    \label{eq:za_variance}
\end{equation}

Since the random term in \eqref{eq:approximate_error_event} is a real-valued linear combination of Gaussian noise samples, $Z_a$ conditioned on $\mathbf{H}$ and $\mathbf{S}_{k-1}$ follows the Gaussian distribution with the mean and variance given in \eqref{eq:za_mean} and \eqref{eq:za_variance}. Next, we define the received squared distance between the two information matrices as
\setlength{\abovedisplayskip}{4pt}
\setlength{\belowdisplayskip}{4pt}
\begin{equation}
    \xi
    \triangleq
    \left\|
        \mathbf{H}\mathbf{S}_{k-1}\mathbf{\Delta}
    \right\|_F^2.
    \label{eq:effective_distance}
\end{equation}
Subsequently, using the nominal SNR in \eqref{eq:snr_definition}, the conditional PEP (CPEP) associated with the high-SNR approximation is 
\setlength{\abovedisplayskip}{4pt}
\setlength{\belowdisplayskip}{4pt}
\begin{equation}
\begin{aligned}
    P_a
    \left(
        \mathbf{X}\rightarrow\hat{\mathbf{X}}
        \mid
        \mathbf{H},\mathbf{S}_{k-1}
    \right)
    &=
    P_a
    \left(
        Z_a<0
        \mid
        \mathbf{H},\mathbf{S}_{k-1}
    \right) \\
    &=
    Q\left( \frac{\mu_Z}{\sigma_Z} \right) = 
    Q\left(
        \sqrt{\frac{\rho\xi}{4}}
    \right).
\end{aligned}
\label{eq:cond_pep}
\end{equation}

\vspace{-0.5em}
\subsection{Channel- and State-Averaged PEP and BEP Analysis}
\label{subsec:channel_averaged_pep}

The CPEP in \eqref{eq:cond_pep} is conditioned on the channel realization $\mathbf{H}$ and the previous differential state $\mathbf{S}_{k-1}$. We first average over $\mathbf{H}$ to obtain the PEP for a given previous state. The resulting PEPs are then averaged over the possible previous differential states to obtain the analytical BEP union bound.

To average the CPEP over the Rician channel, we employ Craig's representation of the $Q$-function \cite{258319}: 

\begin{equation} 
    Q(x) = \frac{1}{\pi} \int_{0}^{\pi/2} \exp\left( -\frac{x^2}{2\sin^2\theta} \right) d\theta. 
    \label{eq:craig_representation} 
\end{equation} 
Substituting \eqref{eq:craig_representation} into \eqref{eq:cond_pep} and taking the expectation over $\mathbf{H}$ gives \vspace{-0.5em}
\setlength{\abovedisplayskip}{4pt}
\setlength{\belowdisplayskip}{4pt}
\begin{equation} 
\begin{aligned} 
    &P_{\mathrm a} \left( \mathbf{X}\rightarrow\hat{\mathbf{X}} \mid\mathbf{S}_{k-1} \right)= \frac{1}{\pi} \int_{0}^{\pi/2} M_{\xi\mid\mathbf{S}_{k-1}} \left( -\frac{\rho}{8\sin^2\theta} \right) d\theta, 
\end{aligned} 
\label{eq:mgf_pep} 
\end{equation}
where $M_{\xi\mid\mathbf{S}_{k-1}}(\nu) \triangleq \mathbb{E}_{\mathbf{H}} [\exp(\nu\xi)\mid\mathbf{S}_{k-1}]$ denotes the MGF of $\xi$ conditioned on the previous differential state $\mathbf S_{k-1}$ \cite{simon2005digital}, and $\nu\in\mathbb{R}$ is the MGF argument.

\subsubsection{Conditional MGF Under Normalized Rician Fading}

To evaluate the MGF in \eqref{eq:mgf_pep}, we adopt a normalized i.i.d. Rician model in which all links have unit average power and share the same Rician factor $K$. The channel matrix is decomposed as
\looseness=-1
\setlength{\abovedisplayskip}{5pt}
\setlength{\belowdisplayskip}{5pt}
\begin{equation}
    \mathbf{H} = \bar{\mathbf{H}} + \widetilde{\mathbf{H}}, 
    \label{eq:analytical_rician_decomposition} 
\end{equation}
where $\bar{\mathbf{H}}$ is the deterministic LoS mean and $\widetilde{\mathbf{H}}$ is the scattered NLoS component. To obtain a tractable evaluation of the analytical expressions, the LoS mean is assumed homogeneous as \looseness=-1
\begin{equation} 
    \bar{\mathbf{H}} = \sqrt{\frac{K}{K+1}} \mathbf{1}_{N_r\times N_t}, 
    \label{eq:homogeneous_rician_mean} 
\end{equation} 
while the scattered entries are independently distributed as %
\setlength{\abovedisplayskip}{4pt}
\setlength{\belowdisplayskip}{4pt}
\begin{equation}
    [\widetilde{\mathbf{H}}]_{i,j} \sim \mathcal{CN}(0,\sigma_h^2), \qquad \sigma_h^2 = \frac{1}{K+1}. 
    \label{eq:rician_nlos_components} 
\end{equation}
This normalization ensures $\mathbb{E}[|h_{i,j}|^2]=1$. Under this model, the pairwise error probability depends on the difference-matrix eigenvalues, the number of receive antennas, the scattered-channel variance, and the deterministic LoS energy along each pairwise error direction.

\begin{lemma} 
\label{lem:mgf} 
Under the Rician model in \eqref{eq:analytical_rician_decomposition}--\eqref{eq:rician_nlos_components}, define
\setlength{\abovedisplayskip}{4pt}
\setlength{\belowdisplayskip}{4pt}
\begin{equation} 
    \mathbf{\Delta}_{\rm eff} \triangleq \mathbf{S}_{k-1}\mathbf{\Delta} \in\mathbb{C}^{N_t\times N_t}. 
    \label{eq:effective_difference_matrix} 
\end{equation}
Let \vspace{-0.5em}
\setlength{\abovedisplayskip}{4pt}
\setlength{\belowdisplayskip}{4pt}
\begin{equation} 
    \mathbf{\Delta}_{\rm eff} \mathbf{\Delta}_{\rm eff}^{H} = \mathbf{V}\mathbf{\Lambda}\mathbf{V}^{H} 
    \label{eq:effective_difference_evd} 
\end{equation}
be its eigendecomposition, where $\mathbf{V}=[\mathbf{v}_1,\ldots,\mathbf{v}_{N_t}]$ is unitary and $\mathbf{\Lambda} =\operatorname{diag} (\lambda_1,\ldots,\lambda_r,0,\ldots,0)$. Here, $r=\operatorname{rank}(\mathbf{\Delta}_{\rm eff})$ and $\lambda_i>0$, $i=1,\ldots,r$, are the nonzero eigenvalues.
For $\nu<1/(\lambda_{\max}\sigma_h^2)$, where $\lambda_{\max}=\mathop{\mathrm{max}}\limits_{1\leq i\leq r}\lambda_i$, the conditional MGF of $\xi$ is
\setlength{\abovedisplayskip}{4pt}
\setlength{\belowdisplayskip}{4pt}
\begin{equation}
\begin{aligned} 
    M_{\xi\mid\mathbf{S}_{k-1}}(\nu) &= \prod_{i=1}^{r} \Bigg[ \left( 1-\nu\lambda_i\sigma_h^2 \right)^{-N_r} \exp\left( \frac{ \nu\lambda_i \mu_i(\mathbf{S}_{k-1}) }{ 1-\nu\lambda_i\sigma_h^2 } \right) \Bigg],
\end{aligned} 
\label{eq:mgf_exact} 
\end{equation}
where\vspace{-0.5em}
\setlength{\abovedisplayskip}{4pt}
\setlength{\belowdisplayskip}{4pt}
\begin{equation} 
    \mu_i(\mathbf{S}_{k-1}) \triangleq \left\| \bar{\mathbf{H}}\mathbf{v}_i \right\|^2 
    \label{eq:noncentrality_effective} 
\end{equation}
represents the deterministic LoS energy along the $i$th pairwise error direction, and captures the contribution of the LoS component to the conditional MGF. 

\end{lemma}
\begin{proof} See Appendix~\ref{app:lemma_mgf}. \end{proof}

\begin{remark}[Impact of the Previous Differential State]
Lemma~\ref{lem:mgf} reveals that the previous differential state affects the conditional MGF only through the LoS-dependent term $\mu_i(\mathbf{S}_{k-1})$. Therefore, under the adopted i.i.d. Rician model, the previous state influences the error performance by changing the alignment between the LoS channel and the error directions. In the absence of the LoS component, this dependence disappears.
\end{remark}

\subsubsection{State-Averaged PEP and BEP}

Averaging the conditional PEP over the Rician channel using the MGF expression in \eqref{eq:mgf_exact} gives the PEP conditioned only on the previous differential state. Defining $a_{\theta}\triangleq\rho/(8\sin^2\theta)$, we obtain
\setlength{\abovedisplayskip}{4pt}
\setlength{\belowdisplayskip}{4pt}
\begin{equation} 
\begin{aligned}
    P_{\mathrm a} \left( \mathbf{X}\rightarrow\hat{\mathbf{X}}  \mid\mathbf{S}_{k-1} \right) & =  \frac{1}{\pi} \int_{0}^{\pi/2} \prod_{i=1}^{r} \Bigg[ \left( 1+a_{\theta}\lambda_i\sigma_h^2 \right)^{-N_r}\\ &\qquad\times \exp\left( -\frac{ a_{\theta}\lambda_i \mu_i(\mathbf{S}_{k-1}) }{ 1+a_{\theta}\lambda_i\sigma_h^2 } \right) \Bigg] d\theta. 
\end{aligned}
\label{eq:explicit_channel_averaged_pep} 
\end{equation}

This expression still depends on the previous differential state $\mathbf S_{k-1}$, which is determined by the transmitted information history. Therefore, the overall PEP is obtained by averaging over all possible previous states. Let $\mathcal{T}$ denote the set of possible previous differential states. The resulting state-averaged PEP is

\begin{equation}\label{35}
    \bar P_{\mathrm a} (\mathbf X\rightarrow\hat{\mathbf X}) = \frac{1}{|\mathcal T|} \sum_{\mathbf S_{k-1}\in\mathcal T} P_{\mathrm a} (\mathbf X\rightarrow\hat{\mathbf X}|\mathbf S_{k-1}).
\end{equation}

\begin{lemma}
\label{lem:perm_invariance}
Under the homogeneous LoS model in \eqref{eq:homogeneous_rician_mean}, let the previous differential state be decomposed as
\begin{equation}
    \mathbf S_{k-1}=\mathbf P_{k-1}\mathbf\Phi_{k-1},
\end{equation} 
where $\mathbf P_{k-1}$ is the accumulated permutation matrix and $\mathbf\Phi_{k-1}$ is a diagonal accumulated PSK phase matrix. Then, the state-conditioned PEP in \eqref{eq:explicit_channel_averaged_pep} is invariant to the permutation component,
i.e.,
\begin{equation}
    P_{\mathrm a}(\mathbf X\rightarrow\hat{\mathbf X} |\mathbf S_{k-1}) = P_{\mathrm a}(\mathbf X\rightarrow\hat{\mathbf X} |\mathbf\Phi_{k-1}).
    \label{eq:permutation_invariant_pep}
\end{equation}
\end{lemma}

\begin{proof}
    See Appendix \ref{app:state_reduction}.
\end{proof}
This invariance relies on $\bar{\mathbf H}\mathbf P_{k-1}=\bar{\mathbf H}$ and need not hold for a geometry-dependent LoS mean.

Consequently, the averaging over all possible previous differential states can be reduced to an averaging over the set of possible accumulated phase matrices. Let $\mathcal F$ denote this set. The state-averaged PEP becomes
\setlength{\abovedisplayskip}{4pt}
\setlength{\belowdisplayskip}{4pt}
\begin{equation} 
    \begin{aligned} 
    \bar{P}_{\mathrm a} \left( \mathbf{X}\rightarrow\hat{\mathbf{X}} \right) 
    &\triangleq \frac{1}{|\mathcal{F}|} \sum_{\mathbf{\Phi_{k-1}}\in\mathcal{F}} P_{\mathrm a} \left( \mathbf{X}\rightarrow\hat{\mathbf{X}} \mid\mathbf{\Phi}_{k-1}\right).
    \end{aligned}
\label{eq:implemented_state_average} 
\end{equation}
Since each information matrix is given by $\mathbf X_k=\mathbf A_{q_k}\mathbf D_k$, where $\mathbf D_k$ contains the PSK symbols, the accumulated phase matrix $\mathbf\Phi_{k-1}$ remains diagonal with PSK entries. This follows from the closure property of the PSK constellation under multiplication, i.e., $e^{j2\pi m_1/M}e^{j2\pi m_2/M} = e^{j2\pi(m_1+m_2)/M}$. Therefore, the set $\mathcal F$ contains all possible diagonal accumulated phase matrices, whose diagonal entries are selected from the underlying $M$-PSK constellation. Hence,
\begin{equation}
    |\mathcal F|=M^{N_t}.
\end{equation}
Without this reduction, the state average would require enumerating both the accumulated permutation and phase components. Lemma~\ref{lem:perm_invariance} removes the explicit permutation dimension from this calculation. 
 
To obtain the BEP over the complete codebook, we average the pairwise error events over all transmitted and erroneously detected information matrices and weight each event by its corresponding bit error distance:
\setlength{\abovedisplayskip}{4pt}
\setlength{\belowdisplayskip}{4pt}
\begin{equation} 
\begin{aligned} 
    P_{b,\mathrm a} \leq \frac{1}{B|\mathcal{C}|} \sum_{\mathbf{X}\in\mathcal{C}} \sum_{\substack{ \hat{\mathbf{X}}\in\mathcal{C}\\ \hat{\mathbf{X}}\neq\mathbf{X} }} &d_H(\mathbf{X},\hat{\mathbf{X}}) \bar{P}_{\mathrm a} \left( \mathbf{X}\rightarrow\hat{\mathbf{X}} \right),
\end{aligned} 
\label{eq:general_analytical_union_bound} 
\end{equation} 
where $B$ denotes the number of bits transmitted per block, $|\mathcal C|$ is the codebook size, and $d_H(\mathbf X,\hat{\mathbf X})$ represents the number of bit errors caused by detecting $\hat{\mathbf X}$ instead of $\mathbf X$.

The derived expressions apply to the normalized i.i.d. common-$K$ Rician model and do not capture the link-dependent path loss, blockage, Rician factors, waveguide attenuation, or geometry-dependent LoS structure of the PA channel. In particular, the phase-state reduction in Lemma~\ref{lem:perm_invariance} relies on the homogeneous LoS mean.\vspace{-0.5em}

\subsection{Asymptotic Pairwise Error and Codebook Analysis}
\label{subsec:asymptotic_analysis}

The state-averaged PEP in \eqref{eq:implemented_state_average} averages over the representative phase matrices in $\mathcal F$. Based on Lemma~\ref{lem:perm_invariance}, we first characterize the high-SNR behavior of the PEP for a fixed $\mathbf\Phi_{k-1}\in\mathcal F$ and then average the result over $\mathcal F$.

\begin{theorem}[Asymptotic Pairwise Error Upper Bound] \label{thm:asymptotic}

For a given pairwise error event $\mathbf X\rightarrow\hat{\mathbf X}$ and representative previous phase state $\mathbf\Phi_{k-1}\in\mathcal F$, the corresponding analytical PEP satisfies
\begin{equation} 
\begin{aligned} 
    &P_{\mathrm a} \left( \mathbf X\rightarrow\hat{\mathbf X} \mid\mathbf\Phi_{k-1} \right)\leq \frac{1}{2} \left( G_c\rho \right)^{-G_d} e^{-\Gamma_{\rm PA}(\mathbf\Phi_{k-1})}[ 1+o(1)], 
\end{aligned} 
\label{eq:asymptotic_pep} 
\end{equation}
as the nominal SNR $\rho\rightarrow\infty$, where $o(1)\rightarrow0$. The polynomial SNR exponent is 
\begin{equation}
    G_d \triangleq rN_r, \qquad r \triangleq \operatorname{rank}(\mathbf{\Delta}_{\rm eff}) = \operatorname{rank}(\mathbf{\Delta}), 
    \label{eq:pairwise_diversity} 
\end{equation}
the pairwise coding coefficient is
\setlength{\abovedisplayskip}{4pt}
\setlength{\belowdisplayskip}{4pt}
\begin{equation} 
    G_c \triangleq \frac{1}{8} \left( \prod_{i=1}^{r}\lambda_i \right)^{1/r} \sigma_h^2 = \frac{1}{8(K+1)} \left( \prod_{i=1}^{r}\lambda_i \right)^{1/r}, 
    \label{eq:coding_gain_final} 
\end{equation}
and the phase-dependent LoS contribution is
\begin{equation} 
    \Gamma_{\rm PA}(\mathbf\Phi_{k-1}) \triangleq \sum_{i=1}^{r} \frac{ \mu_i(\mathbf\Phi_{k-1}) }{ \sigma_h^2 }. 
    \label{eq:pa_decay_general} 
\end{equation}
Here, $\lambda_i$, $i=1,\ldots,r$, are the nonzero eigenvalues of $\mathbf{\Delta}\mathbf{\Delta}^{H}$, and $\mathbf{u}_i$ is a corresponding unit-norm eigenvector. Moreover, $\mu_i(\mathbf{\Phi}_{k-1}) \triangleq \|\bar{\mathbf H}\mathbf{\Phi}_{k-1}\mathbf{u}_i\|^2$ is the deterministic LoS energy along the $i$th pairwise error direction. This phase-only representation follows from the permutation invariance established in Lemma~\ref{lem:perm_invariance}.

Averaging \eqref{eq:asymptotic_pep} over $\mathbf\Phi\in\mathcal F$ gives
\setlength{\abovedisplayskip}{4pt}
\setlength{\belowdisplayskip}{4pt}
\begin{equation} 
\begin{aligned} 
    \bar P_{\mathrm a} \left( \mathbf X\rightarrow\hat{\mathbf X} \right) &\leq \frac{(G_c\rho)^{-G_d}}{2|\mathcal F|} \sum_{\mathbf\Phi\in\mathcal F} e^{-\Gamma_{\rm PA}(\mathbf\Phi)} [1+o(1)]. 
\end{aligned} 
\label{eq:state_averaged_asymptotic_pep} 
\end{equation}
Thus, phase-state averaging preserves $G_d$ and $G_c$ and only averages the LoS-dependent factor.

\end{theorem}

\begin{proof} See Appendix~\ref{app:asymptotic_pep}. \end{proof}

Theorem~\ref{thm:asymptotic} characterizes one pairwise error event. The complete BEP union bound in \eqref{eq:general_analytical_union_bound}, however, contains terms for all distinct information-matrix pairs. For each pair, the difference rank $r$ determines the polynomial SNR exponent $G_d=rN_r$. Hence, pairwise terms with smaller rank decrease more slowly within the asymptotic upper bound. The minimum difference rank therefore identifies the slowest-decaying terms in the analytical BEP.

The numerical analytical evaluation reported later fixes one transmitted reference matrix rather than averaging over every possible transmitted matrix. We therefore compare the minimum difference rank of the complete codebook with the minimum difference rank among errors originating from the selected reference matrix, $\mathbf{X}_{\rm ref}$.

\begin{corollary}[Codebook- and Reference-Level Minimum Exponents]
\label{cor:codebook_diversity}

For the employed PA-DSM information-matrix codebook, which contains all $M^{N_t}$ PSK-symbol combinations for each retained permutation, define the minimum pairwise difference rank over the complete codebook as
\setlength{\abovedisplayskip}{4pt}
\setlength{\belowdisplayskip}{4pt}
\begin{equation}
\begin{aligned}
    r_{\min} &\triangleq \min_{\substack{ \mathbf{X},\hat{\mathbf{X}}\in\mathcal{C}\\ \mathbf{X}\neq\hat{\mathbf{X}} }} \operatorname{rank} \left( \mathbf{X}-\hat{\mathbf{X}} \right),\\ r_{\min}^{\rm ref} 
    &\triangleq \min_{\substack{ \hat{\mathbf{X}}\in\mathcal{C}\\ \hat{\mathbf{X}}\neq\mathbf{X}_{\rm ref} }} \operatorname{rank} \left( \mathbf{X}_{\rm ref}-\hat{\mathbf{X}} \right). 
\end{aligned} 
\label{eq:minimum_pairwise_ranks}
\end{equation} 
Then, 
\begin{equation} 
    r_{\min}^{\rm ref} = r_{\min} = 1, \qquad G_{d,\min}^{\rm ref} = G_{d,\min} = N_r, 
    \label{eq:pa_dsm_aber_diversity} 
\end{equation} 
where $G_{d,\min}^{\rm ref} =r_{\min}^{\rm ref}N_r$ and $G_{d,\min}=r_{\min}N_r$. 
\end{corollary}

\begin{proof}
Consider any $\mathbf{X}=\mathbf{A}_q\mathbf{D}_s\in\mathcal{C}$, including $\mathbf{X}=\mathbf{X}_{\rm ref}$. Since the codebook contains every PSK-symbol combination for each retained permutation, it contains another information matrix $\hat{\mathbf{X}} =\mathbf{A}_q\hat{\mathbf{D}}_s$ such that $\mathbf{D}_s$ and $\hat{\mathbf{D}}_s$ differ in exactly one diagonal entry. Therefore,
\setlength{\abovedisplayskip}{4pt}
\setlength{\belowdisplayskip}{4pt}
\begin{equation}
    \operatorname{rank} \left( \mathbf{X}-\hat{\mathbf{X}} \right) = \operatorname{rank} \left[ \mathbf{A}_q \left( \mathbf{D}_s-\hat{\mathbf{D}}_s \right) \right] = 1, 
    \label{eq:rank_one_pair}
\end{equation}
because $\mathbf{A}_q$ is nonsingular and $\mathbf{D}_s-\hat{\mathbf{D}}_s$ has exactly one nonzero diagonal entry. Thus, a rank-one difference exists for every transmitted information matrix, including any selected reference matrix, which gives $r_{\min}\leq1$ and $r_{\min}^{\rm ref}\leq1$. Since the difference between two distinct information matrices cannot have rank zero, both minimum ranks equal one. Substitution into $G_d=rN_r$ yields \eqref{eq:pa_dsm_aber_diversity}.
\end{proof}

Evaluating \eqref{eq:general_analytical_union_bound} requires treating every matrix in $\mathcal C$ as transmitted. To reduce this numerical burden, the analytical curves in Section~\ref{subsec:antenna_scaling_canyon} fix $\mathbf X_{\rm ref}\in\mathcal C$ and evaluate
\begin{equation} 
\begin{aligned} 
    P_{b,\mathrm a}^{\rm ref} \triangleq \frac{1}{B} \sum_{\substack{ \hat{\mathbf{X}}\in\mathcal{C}\\ \hat{\mathbf{X}}\neq\mathbf{X}_{\rm ref} }} &d_H \left( \mathbf{X}_{\rm ref}, \hat{\mathbf{X}} \right) \bar{P}_{\mathrm a} \left( \mathbf{X}_{\rm ref} \rightarrow \hat{\mathbf{X}} \right).
\end{aligned}
\label{eq:reference_codeword_characterization} 
\end{equation}

Corollary~\ref{cor:codebook_diversity} establishes that the reference-codeword sum in \eqref{eq:reference_codeword_characterization} and the complete codeword-averaged union bound contain terms with the same minimum polynomial SNR exponent $N_r$. This result does not imply equal finite-SNR values or equal asymptotic coefficients, because the numbers of bits in error, pairwise eigenvalues, and deterministic LoS energy terms can vary with the transmitted information matrix. Therefore, the reference-codeword curves in Fig.~\ref{fig:theory_vs_sim} are used to illustrate the predicted increase in the high-SNR slope with $N_r$, rather than to approximate the complete finite-SNR BEP bound.

We finally examine the remaining terms in \eqref{eq:asymptotic_pep}. Since the previous phase state is unitary, it does not change the nonzero eigenvalues of $\mathbf{\Delta}\mathbf{\Delta}^{H}$. Hence, the pairwise coding coefficient $G_c$ is independent of $\mathbf{\Phi}_{k-1}$. By contrast, the LoS contribution can remain phase dependent because $\mathbf{\Phi}_{k-1}$ changes the alignment of the pairwise error directions $\mathbf u_i$ with the deterministic LoS channel.

\begin{corollary}[LoS Contribution Under the Homogeneous Mean]
\label{cor:homogeneous_los_decay}

Under the homogeneous LoS mean in \eqref{eq:homogeneous_rician_mean}, the phase-dependent LoS contribution in \eqref{eq:pa_decay_general} becomes
\setlength{\abovedisplayskip}{4pt}
\setlength{\belowdisplayskip}{4pt}
\begin{equation}
    \Gamma_{\rm PA}(\mathbf{\Phi}_{k-1}) = K \sum_{i=1}^{r} \left\| \mathbf{1}_{N_r\times N_t} \mathbf{\Phi}_{k-1}\mathbf{u}_i \right\|^2. 
    \label{eq:gamma_pa_homogeneous} 
\end{equation}
\end{corollary}

\begin{proof} 
Substituting the definition of $\mu_i(\mathbf{\Phi}_{k-1})$ in Theorem~\ref{thm:asymptotic} together with \eqref{eq:homogeneous_rician_mean} and \eqref{eq:rician_nlos_components} into \eqref{eq:pa_decay_general} yields \eqref{eq:gamma_pa_homogeneous}.
\end{proof}

The quantities $G_c$ and $\Gamma_{\rm PA}(\mathbf{\Phi}_{k-1})$ describe the scattered NLoS and deterministic LoS contributions to the asymptotic pairwise error upper bound, respectively. Increasing $K$ reduces $\sigma_h^2=1/(K+1)$ and hence $G_c$, which increases the polynomial factor $(G_c\rho)^{-G_d}$. Conversely, a larger $K$ strengthens the deterministic LoS contribution and reduces the exponential factor in \eqref{eq:asymptotic_pep}, with the reduction determined by the alignment of $\mathbf{\Phi}_{k-1}\mathbf u_i$ with the homogeneous LoS channel. Thus, the effect of $K$ cannot be inferred from the scattered-channel variance alone. \looseness=-1

\vspace{-0.4em}
\section{Numerical Results and Complexity Assessment}
\label{sec:results}

This section evaluates PA-DSM under two complementary channel settings. First, geometry-dependent Monte Carlo computer simulations assess its performance in representative indoor-office, indoor-factory, and street-canyon FR3 deployments at $16.95$ GHz \cite{10901735,11160744,11161884}. Second, Section~\ref{subsec:antenna_scaling_canyon} compares the Monte Carlo results with the reference-codeword MGF characterization developed in Section~\ref{sec:theoretical_analysis}.

\begin{figure}
    \centering
    \includegraphics[width=0.8\linewidth]{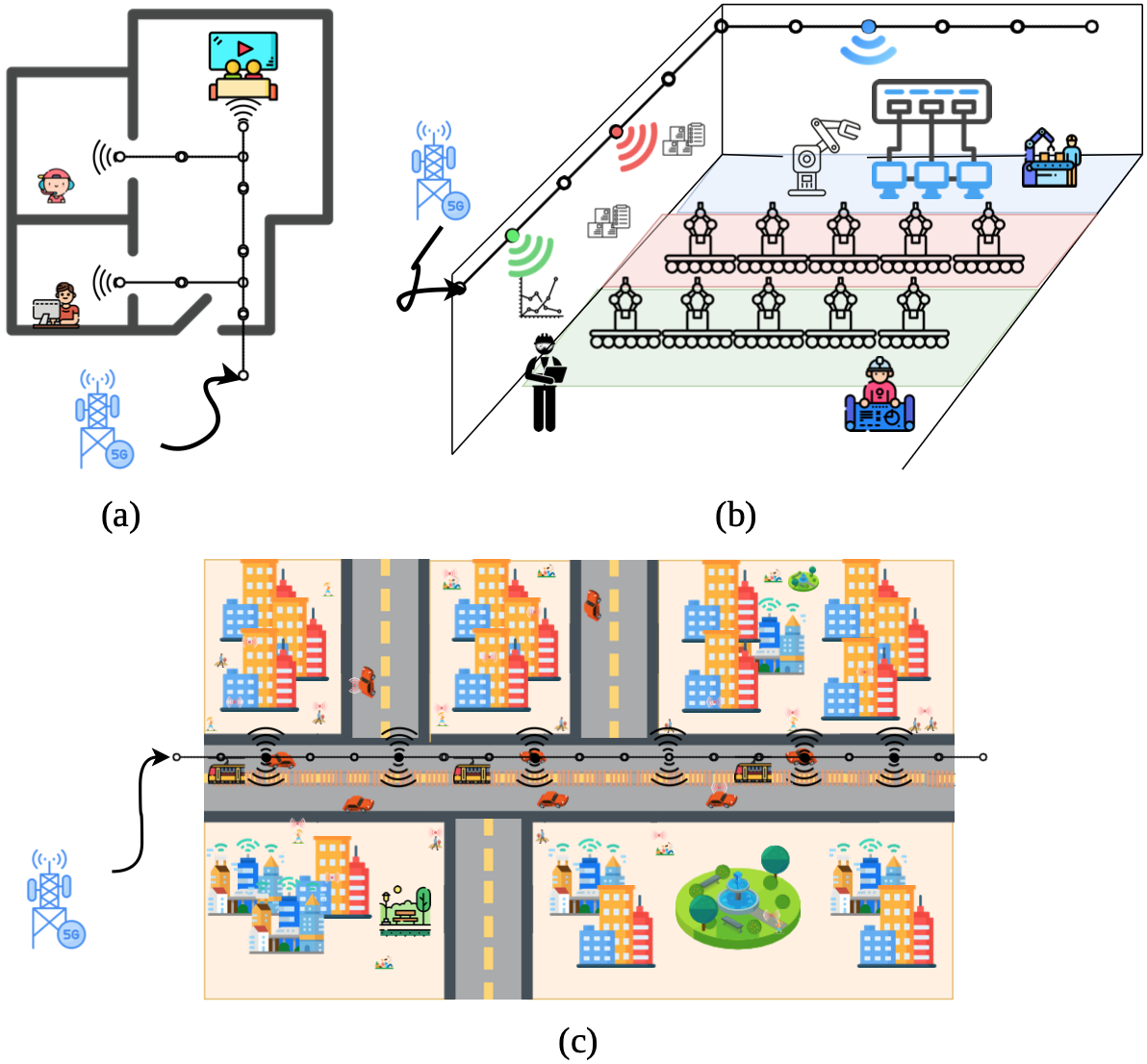}
    \caption{Conceptual layouts of the evaluated FR3 environments: (a) indoor office, (b) indoor factory, and (c) street canyon.}
    \label{fig:env_illustrations}
    
\end{figure}

\vspace{-0.5em}
\subsection{Simulation Environments and Parameters}

Fig.~\ref{fig:env_illustrations} illustrates the three evaluated propagation environments, whose parameters are summarized in Table~\ref{tab:sim_params}. In each environment, the $N_{\rm all}$ candidate PAs are uniformly placed along an $x$-directed waveguide as $\mathbf p_j=[x_j,D_y/2,z_{\rm wg}]^{\mathsf T}$, where $x_j=(j-1)D_x/(N_{\rm all}-1)$, with the feed at $\mathbf u_f=[0,D_y/2,z_{\rm wg}]^{\mathsf T}$. The UE employs an $x$-directed uniform linear array (ULA) centered at $\mathbf u=[u_x,u_y,u_z]^{\mathsf T}$, with half-wavelength spacing and $u_z=1.5$ m. For each evaluated UE location, the $N_t$ PAs are selected once, held fixed during the corresponding PA-DSM trials, and differential encoding is initialized with $\mathbf S_0=\mathbf I_{N_t}$. Unless otherwise stated, the LoS states, shadowing variables, and scattered channel coefficients are generated independently across PA-to-UE links, and each reported bit error rate (BER) point is simulated until at least $N_{\rm err}=5000$ bit errors are observed or $N_{\max}=10^7$ trials are completed. Within each trial, the same channel realization is retained over the two consecutive transmission blocks required for differential detection, in accordance with the quasi-static assumption in Section~\ref{subsec:signal_model}.

The LoS/NLoS state of each link is drawn according to a distance-dependent LoS probability, modeled for the simulations as \cite{10901735,8620255}
\begin{equation} 
    P_{{\rm LoS},i,j} = 
    \begin{cases} 1-\dfrac{d_{{\rm free},i,j}}{d_{\rm cut}}, & 0\leq d_{{\rm free},i,j}\leq d_{\rm cut},\\[1mm] 0, & d_{{\rm free},i,j}>d_{\rm cut},
    \end{cases} 
\label{eq:los_probability} 
\end{equation}
where $d_{\rm cut}$ is given in Table~\ref{tab:sim_params}.

\begin{table}[t]
\centering
\caption{Simulation parameters}
\label{tab:sim_params}
\renewcommand{\arraystretch}{1.0}
\setlength{\tabcolsep}{2.5pt}
\footnotesize
\begin{tabularx}{\columnwidth}{@{}Xccc@{}}
\toprule
\textbf{Global Parameter} & \multicolumn{3}{c}{\textbf{Value}}\\
\midrule
Carrier frequency / bandwidth, $f_c$ / $W$ \cite{10901735}
    & \multicolumn{3}{c}{16.95 GHz / 50 MHz}\\
Noise power / figure, $\sigma_n^2$ / NF
    & \multicolumn{3}{c}{$-90$ dBm (incl. NF) / 7 dB}\\
UE height / speed, $u_z$ / $v$
    & \multicolumn{3}{c}{1.5 m / 1 m/s}\\
Receive array, $N_r$
    & \multicolumn{3}{c}{$1,2,3$, $x$-directed ULA, $\lambda/2$}\\
Selected PAs, $N_t$
    & \multicolumn{3}{c}{$3,4$}\\
Waveguide attenuation / index, $\alpha$ / $n_{\rm eff}$ \cite{10945421,11368709}
    & \multicolumn{3}{c}{0.01 dB/m / 1.4}\\
Pilot interval / count, $T_p$ / $N_p$
    & \multicolumn{3}{c}{0.5 ms / $N_t$}\\
Beacon samples per candidate, $L_p$
    & \multicolumn{3}{c}{$4$}\\
Modulation, PA-DSM / benchmark
    & \multicolumn{3}{c}{QPSK, 8-PSK / QPSK, 16-QAM}\\
\midrule
\textbf{Environment-Specific Parameter} & \textbf{Office} & \textbf{Factory} & \textbf{Canyon}\\
\midrule
Area, $D_x\times D_y$ [m]                       & $15\times6$ & $36\times22$ & $200\times40$\\
Waveguide height, $z_{\rm wg}$ [m]              & 2.8   & 10.0  & 12.0\\
UE position, $\mathbf{u}$ [m] & $(7,3,1.5)$ & $(18,10,1.5)$ & $(120,18,1.5)$\\
Available PAs, $N_{\rm all}$                    & 16    & 16    & 100\\
Maximum LoS distance, $d_{\rm cut}$ [m]         & 20    & 100   & 200\\
Path-loss exponent, $n_{\rm LoS}/n_{\rm NLoS}$  & 1.32/3.07 & 1.75/2.11 & 1.85/2.59\\
Shadowing, $\sigma_{\rm LoS}/\sigma_{\rm NLoS}$ [dB] & 2.66/9.03 & 3.10/3.29 & 4.05/8.78\\
\bottomrule
\end{tabularx}
\end{table}

To determine the large-scale gain of each link, the propagation loss is modeled by combining the close-in free-space reference term, the environment-dependent distance loss, waveguide attenuation, and shadow fading. The total path loss is \cite{10901735} \looseness=-1
\setlength{\abovedisplayskip}{4pt}
\setlength{\belowdisplayskip}{4pt}
\begin{equation}
\begin{aligned}
    PL_{\text{Total}}[\text{dB}]
    &= 32.4 + 20\log_{10}(f_c)
    + 10 n_m \log_{10}(d_{\text{free,i,j}}) \\
    &\quad + \alpha d_{\text{wg,j}}
    + \chi_{\sigma, m,i,j},
\end{aligned}
\end{equation}
where $f_c$ is in GHz, $d_{{\rm free},i,j}$ and $d_{{\rm wg},j}$ are in meters, and $\alpha$ is the waveguide attenuation coefficient in dB/m. The index $m\in\{{\rm LoS},{\rm NLoS}\}$ denotes the propagation state, and $n_m$ is the path-loss exponent for that state. The shadowing contribution in the dB domain is modeled as \looseness=-1
\setlength{\abovedisplayskip}{4pt}
\setlength{\belowdisplayskip}{4pt}
\begin{equation}
    \chi_{\sigma,m,i,j}
    \sim
    \mathcal{N}\left(0,\sigma_m^2\right),
    \label{eq:shadowing_distribution}
\end{equation} 
where $\sigma_m$ is the environment- and state-dependent shadowing standard deviation \cite{10901735,11368709}.

The small-scale fading is modeled using a Rician distribution where the $K$-factor is distance-dependent in an LoS environment and zero in an NLoS environment \cite{8645135}:
\begin{equation} 
    K_{i,j} =
    \begin{cases} 
    10^{1.3-0.003d_{{\rm free},i,j}}, & \text{for a LoS link},\\ 0, & \text{for an NLoS link}. 
    \end{cases} 
\label{eq:distance_dependent_k} 
\end{equation}

\vspace{-0.5em}

\subsection{Performance Evaluation Across Diverse FR3 Environments}
\label{subsec:environment_comparison} 

To evaluate the sensitivity of PA-DSM to the propagation environment, Fig.~\ref{fig:aber_all_envs} compares its BER in indoor-office, indoor-factory, and street-canyon configurations illustrated in Fig.~\ref{fig:env_illustrations}. The three configurations use the environment-specific area dimensions, waveguide heights, LoS probabilities, path-loss exponents, and shadowing parameters summarized in Table~\ref{tab:sim_params}. The noise variance $\sigma_n^2$ is fixed for these simulations. Thus, increasing the transmit power $P_t$ increases the nominal transmit SNR $\rho=P_t/\sigma_n^2$, while the geometry-dependent channel gain determines the received signal level at the UE.

At relatively low transmit powers, the indoor-office configuration provides the lowest BER among the considered environments. This behavior is consistent with the shorter propagation distances and the corresponding large-scale channel gains in the considered office layout. As $P_t$ increases, however, the relative ordering of the environmental curves changes: the indoor-factory configuration continues to exhibit a decreasing BER and provides the lowest BER in the higher-power portion of the evaluated range. The street-canyon configuration lies between the two indoor configurations over most of the considered transmit power values.

\begin{figure}[t]
    \centering
      \includegraphics[width=0.9\columnwidth]{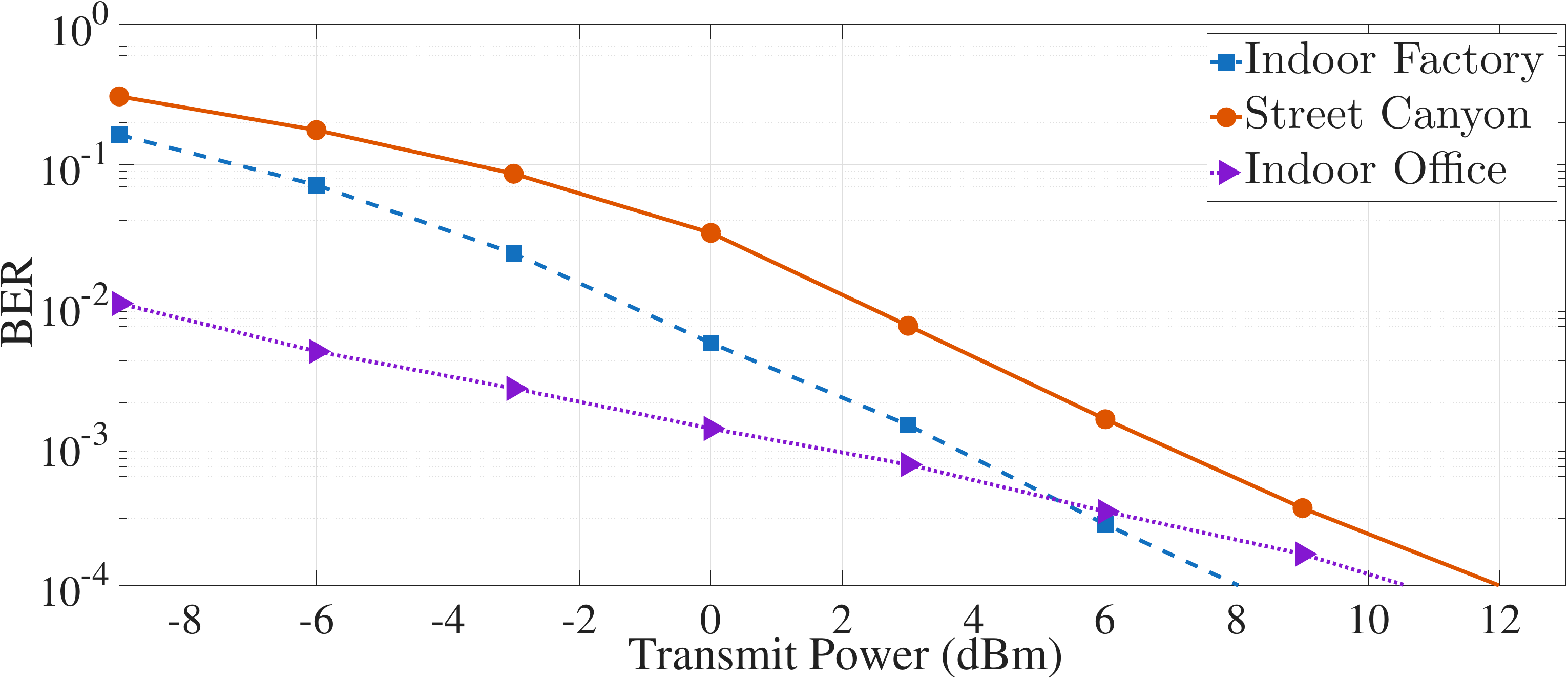}
    \caption{BER performance of PA-DSM in the indoor-office, indoor-factory, and street-canyon environments at $\eta=4$ bps/Hz with $N_t=4$.}
    \label{fig:aber_all_envs}
\end{figure}

The observed crossover indicates that received signal strength alone does not determine the PA-DSM error performance. At higher transmit powers, the distinguishability of the channel responses associated with different PA activation sequences also becomes an important consideration. In the adopted channel model, this distinguishability depends jointly on the selected PA locations, link distances, LoS/NLoS realizations, shadowing, Rician factors, and scattered-channel realizations. The environment-dependent combinations of these quantities can alter the pairwise distances between competing PA-DSM hypotheses and thereby change their relative BER performance. Isolating the specific contribution of spatial correlation or angular spread would require additional channel statistics not evaluated here.\looseness=-1

\vspace{-0.5em}
\subsection{RSSI-Assisted PA Selection Performance} 
\label{subsec:spatial_coverage} 
\vspace{-0.3em}



Fig.~\ref{fig:user_loc_aber} compares the BER versus the longitudinal UE position in the street canyon. All configurations activate $N_t=4$ PAs from the same $N_{\rm all}=100$ candidates: the proposed rule takes the four strongest RSSI values, the nearest-PA baseline takes the four smallest UE distances and hence assumes location knowledge, and the fixed baseline uses four uniformly spaced PAs without adaptation.

Both adaptive rules keep the BER in the $10^{-5}$ range over the full $200$~m span, whereas the fixed subset stays above $7\times10^{-2}$ everywhere; since a DSM block uses all $N_t$ PAs, its performance is set by the weakest active link, which is why the fixed curve degrades at both ends of the corridor. The proposed rule further outperforms nearest-PA selection by a factor of $2$ to $3$ in BER, because ranking by distance is blind to shadowing whereas $P_{r,n}$ is measured through the actual PA-waveguide path; relative to genie selection on the exact large-scale coefficients, the mean link-gain loss is $0.09$~dB, against $2.79$~dB for nearest-PA and $13.41$~dB for the fixed subset, and it is achieved from scalar power measurements without UE-location knowledge. As with the capacity- versus distance-based antenna selection distinction in conventional SM, this strength-oriented rule adapts on a slow timescale, whereas directly optimizing hypothesis separation would require a CSI-dependent distance-based criterion \cite{10685086}. 

\begin{figure}[!t] 
    \centering
    \includegraphics[width=0.82\columnwidth]{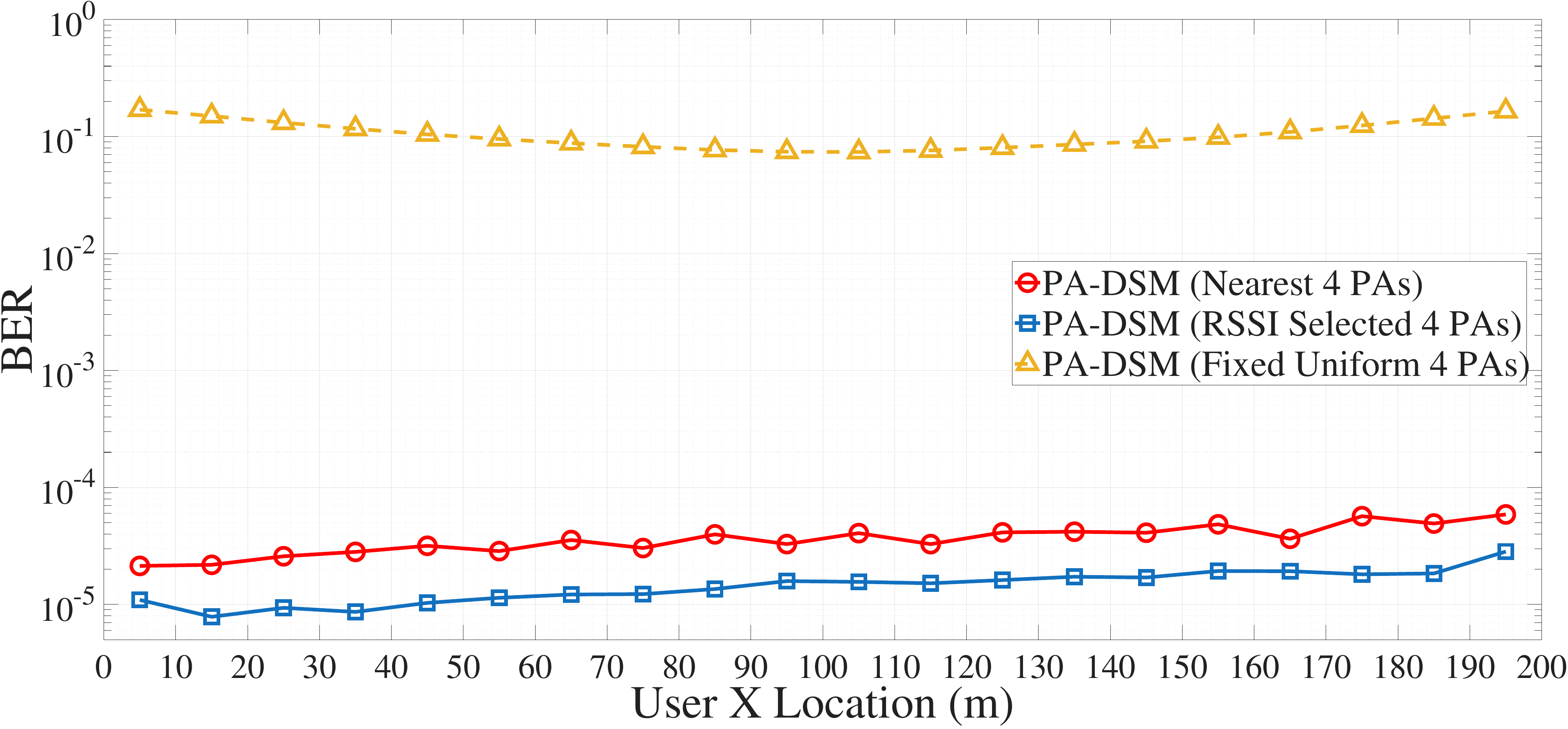}
    \caption{BER versus longitudinal UE position in the street canyon for three $N_t=4$ PA subsets drawn from the same $N_{\rm all}=100$ candidates. $N_r=2$, 8-PSK, $\eta=4$~bps/Hz, $P_t=15$~dBm.}
    \label{fig:user_loc_aber}
\end{figure}

\vspace{-0.5em}
\subsection{Robustness Comparison Against Benchmarks}
\label{subsec:robustness_comparison}

To evaluate the robustness of PA-DSM against representative coherent and non-coherent alternatives, Fig.~\ref{fig:ber_comparison}(a) compares PA-DSM with coherent SM and a coherent single-antenna baseline. All three schemes use the same candidate PA deployment, propagation environment, transmit-power budget, and target spectral efficiency of $\eta=4$ bps/Hz. PA-DSM and coherent SM use the same RSSI-selected subset of $N_t=4$ PAs, thereby avoiding a performance difference caused by using different active subsets. The single-antenna baseline instead transmits from the nearest available PA and does not convey index bits.

At $\eta=4$ bps/Hz, the single-antenna baseline carries all information through 16-ary quadrature amplitude modulation (QAM). Coherent SM uses $N_t=4$ PAs to convey two index bits and quadrature PSK (QPSK) to convey the remaining two bits per channel use. PA-DSM conveys $\lfloor\log_2(N_t!)\rfloor/N_t=1$ bps/Hz through differential permutation indexing and the remaining $3$ bps/Hz through 8-PSK.

As an initial detector-level verification, Fig.~\ref{fig:ber_comparison}(a) shows that the LC-ML and exhaustive-ML curves coincide over the evaluated transmit-power range. This agreement confirms that the PSK-specialized LC-ML detector returns the same minimizer of the differential metric as exhaustive joint search. Under perfect CSI, coherent SM requires approximately $3$ dB less transmit power than PA-DSM at the BER level indicated in the figure. This difference is consistent with the effective noise variance increase associated with differential detection \cite{6879496}. \looseness=-1

To examine sensitivity to channel estimation errors, the coherent benchmarks use the scale-consistent Gauss--Markov model \looseness=-1
\setlength{\abovedisplayskip}{8pt}
\setlength{\belowdisplayskip}{4pt}
\begin{equation}
    [\hat{\mathbf H}]_{i,j} = \varrho_{i,j}[\mathbf H]_{i,j}
    + \sqrt{1-\varrho_{i,j}^{2}}\,[\mathbf E_{\beta}]_{i,j},
    \label{eq:imperfect_csi_model}
\end{equation}
where $\hat{\mathbf H}$ is the channel available to the coherent detector,$\varrho_{i,j}\in[0,1]$ is the per-link estimation-quality coefficient, and $[\mathbf E_{\beta}]_{i,j}\sim \mathcal{CN}(0,\beta_{i,j})$ independently of $\mathbf H$. Thus, the error innovation follows the link-dependent large-scale channel power. Rather than fixing $\varrho_{i,j}$, we relate it to the pilot resources available to the coherent receiver, which two distinct mechanisms limit. The estimation-noise contribution is parameterized by the standard Gaussian minimum mean square error (MMSE) quality factor $\varrho_{\mathrm{est},i,j}^{2} =N_p\rho\beta_{i,j}/(1+N_p\rho\beta_{i,j})$ for $N_p$ pilot symbols at nominal SNR $\rho$, which approaches unity as $P_t$ increases. In addition, an estimate formed once per pilot interval $T_p$ is reused over that interval and therefore ages. Taking $T_p/2$ as the representative pilot-to-data separation, the Clarke temporal correlation is $\zeta=J_0(2\pi f_D T_p/2)=J_0(\pi f_D T_p)$, where $J_0(\cdot)$ is the zeroth-order Bessel function of the first kind and $f_D$ is the maximum Doppler frequency \cite{Clarke}. The effective per-link quality is modeled as $\varrho_{i,j}=\zeta\varrho_{\mathrm{est},i,j}$. Unlike the estimation-noise term, the aging term does not vanish with transmit power. PA-DSM does not incur this pilot-reuse loss, since differential detection references the immediately preceding received block. It remains sensitive to variation between adjacent blocks, but the largest separation between paired columns of $\mathbf Y_k$ and $\mathbf Y_{k-1}$ is $(2N_t-1)/W$ with $W$ the signal bandwidth, for which the correlation loss stays below $10^{-8}$ under the parameters of Table~\ref{tab:sim_params}. We set $N_p=N_t$, corresponding to one pilot symbol per active PA. Perfect CSI corresponds to $\varrho_{i,j}=1$.

\begin{figure}[!t]
    \centering
    \includegraphics[width=0.9\columnwidth]{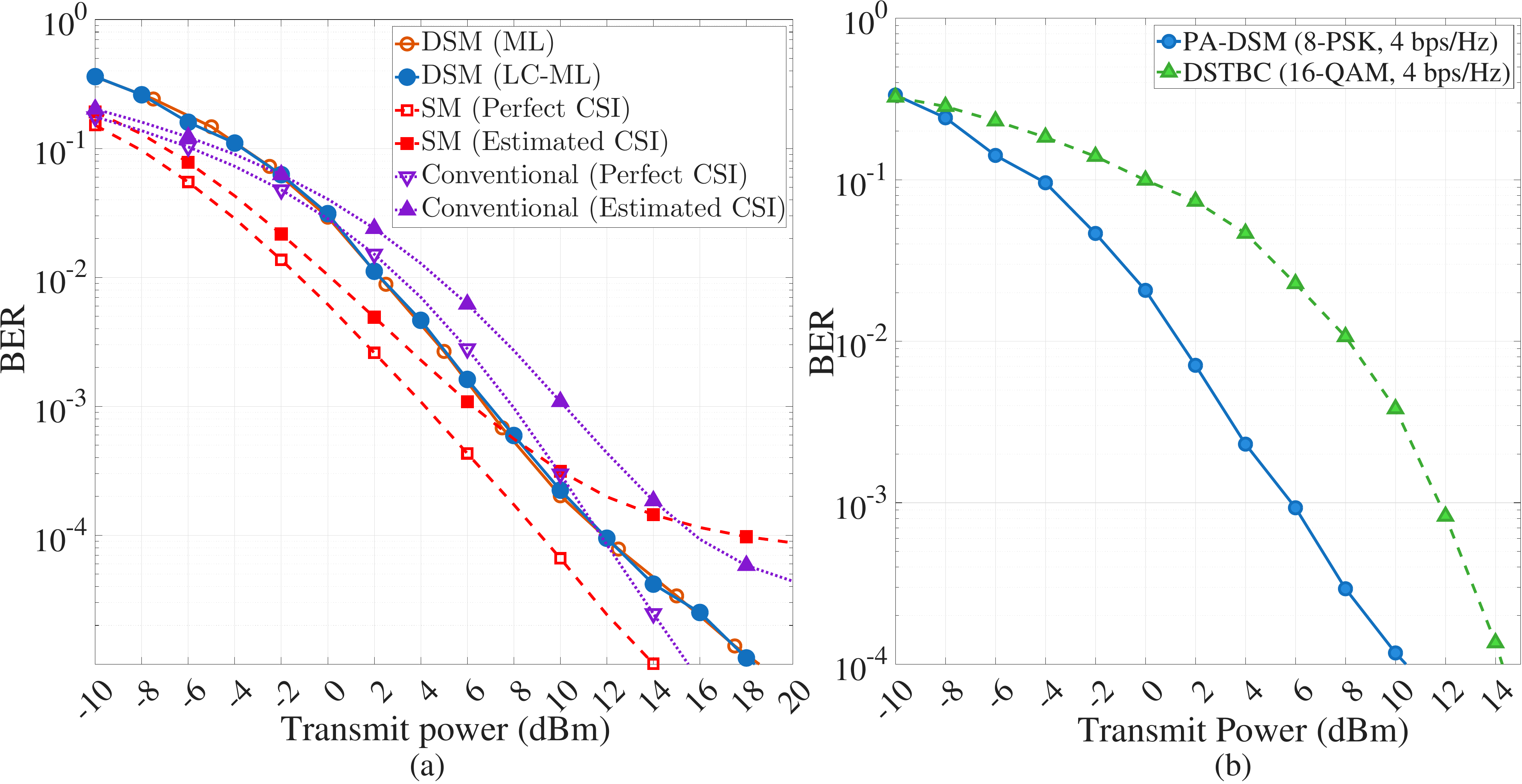}
    \caption{BER comparison at $\eta = 4$~bps/Hz: (a) PA-DSM and coherent
benchmarks under perfect CSI and under MMSE channel estimation from
$N_p = N_t$ pilots reused over a pilot interval $T_p = 0.5$~ms at a UE
speed of 1~m/s; (b) PA-DSM and same-waveguide DSTBC under \eqref{eq:dstbc_depletion} with $\alpha_{\rm e}=0.5$.}
    \label{fig:ber_comparison}    
\end{figure}

Pilot overhead is not deducted from the reported spectral efficiencies of the coherent benchmarks. With $N_p=N_t$ pilots per interval $T_p$, this amounts to $1.6\times10^{-4}$ of the channel uses under the parameters of Table~\ref{tab:sim_params}, so the omission favors the coherent schemes in net data rate only marginally. Under the adopted model, the estimation-noise contribution decreases with transmit power, so the coherent curves do not saturate for that reason. The flattening observed at high transmit power originates instead in channel aging, which caps the effective per-link quality at $\zeta^{2}/(1-\zeta^{2})=24.0$ dB for the same parameters, independently of $P_t$. PA-DSM does not use an explicit instantaneous channel estimate for data detection, and its BER continues to decrease over the same range. The result therefore quantifies a trade between the two detection strategies rather than a universal advantage of either. PA-DSM forgoes the coherent gain observed under perfect CSI in exchange for detection that requires no channel estimate, and it becomes preferable whenever the pilot interval cannot be made short relative to the channel coherence time.\looseness=-1

To compare PA-DSM with a non-coherent transmit-diversity alternative, Fig.~\ref{fig:ber_comparison}(b) presents the BER of PA-DSM and the considered differential space-time block code (DSTBC) benchmark \cite{5672371,965648}. Here, $P_t$ denotes the power supplied at the waveguide feed. The DSTBC benchmark simultaneously activates $L$ PAs on the same series-fed waveguide, with the radiated power of the $\ell$th PA modeled as
\setlength{\abovedisplayskip}{4pt}
\setlength{\belowdisplayskip}{4pt}
\begin{equation}
    P_{\ell} = \alpha_{\rm e}\left(1-\alpha_{\rm e}\right)^{\ell-1}P_t,\qquad \ell=1,\ldots,L,
    \label{eq:dstbc_depletion}
\end{equation}
where $\alpha_{\rm e}$ is the emitting ratio. Thus, its total radiated power is $P_t[1-(1-\alpha_{\rm e})^L]$; for $L=2$ and $\alpha_{\rm e}=0.5$, the two PAs radiate $0.5P_t$ and $0.25P_t$. By contrast, PA-DSM activates one PA per time slot, which is assumed to extract the available feed power. The 16-QAM DSTBC code matrices are normalized to satisfy the same average feed-power constraint as PA-DSM before differential encoding, and the non-coherent detector follows \cite{965648}. Hence, Fig.~\ref{fig:ber_comparison}(b) compares the two schemes under equal feed power while accounting for the depletion caused by simultaneous same-waveguide radiation. Under this specific architecture and power model, PA-DSM achieves the lower BER. Under an equal-radiated-energy convention, the DSTBC curves would shift by $10\log_{10}(1/0.75)=1.25$ dB, which does not alter the observed ordering.

\vspace{-0.5em}
\subsection{Impact of Antenna Scaling Under Rician Fading}
\label{subsec:antenna_scaling_canyon}

\begin{figure}[t]
    \centering
    \includegraphics[width=0.85\linewidth]{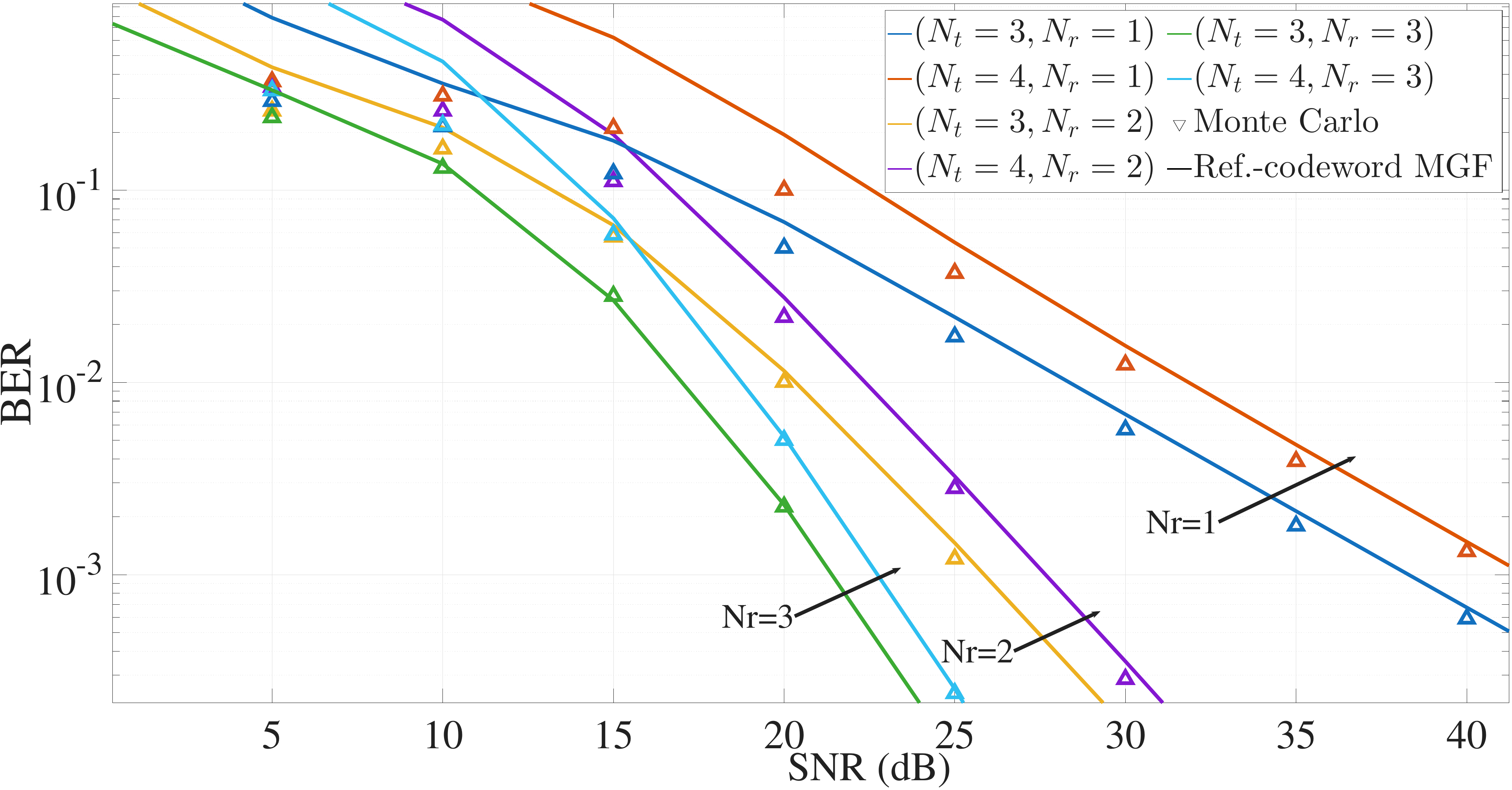}
    \caption{Monte Carlo BER and reference-codeword MGF characterization under normalized common-$K$ Rician fading with $K=12$ for $M=8$.}
    \label{fig:theory_vs_sim}
\end{figure}

Unlike the preceding geometry-dependent evaluations, this subsection isolates the effects of $N_t$ and $N_r$ under the Rician model adopted in Section~\ref{sec:theoretical_analysis}. Fig.~\ref{fig:theory_vs_sim} compares BER results with the MGF-based analytical characterization for $N_t\in\{3,4\}$ and $N_r\in\{1,2,3\}$. Both evaluations use the same Rician channel model. The analytical curves additionally employ the homogeneous LoS mean in \eqref{eq:homogeneous_rician_mean}, uniform averaging over the representative previous phase states in $\mathcal{F}$, and the reference transmitted information matrix defined in \eqref{eq:reference_codeword_characterization}. Accordingly, the analytical curves provide a reference-codeword characterization of the predicted high-SNR slope variation with $N_r$, rather than the complete codeword-averaged finite-SNR BER.

To examine the receive diversity behavior, consider first the curves for different values of $N_r$. Increasing $N_r$ from one to three reduces the BER and produces a steeper high-SNR slope over the evaluated SNR range. This trend is consistent with the pairwise polynomial SNR exponent $G_d=rN_r$ defined in \eqref{eq:pairwise_diversity}. At the codebook level, the employed PA-DSM codebook has $r_{\min}=1$, since two information matrices can share the same spatial permutation and differ in only one PSK symbol. The resulting codebook-level polynomial SNR exponent is therefore $G_{d,\min}=N_r$, as shown in \eqref{eq:pa_dsm_aber_diversity}. Pairwise error events with $r>1$ have larger polynomial SNR exponents, whereas the slowest-decaying terms in the analytical BEP are the rank-one events. Hence, the slope variation observed with $N_r$ is consistent with the receive diversity behavior predicted by the analysis. \looseness=-1

\begin{figure}[!t]
    \centering
    \includegraphics[width=0.9\columnwidth]{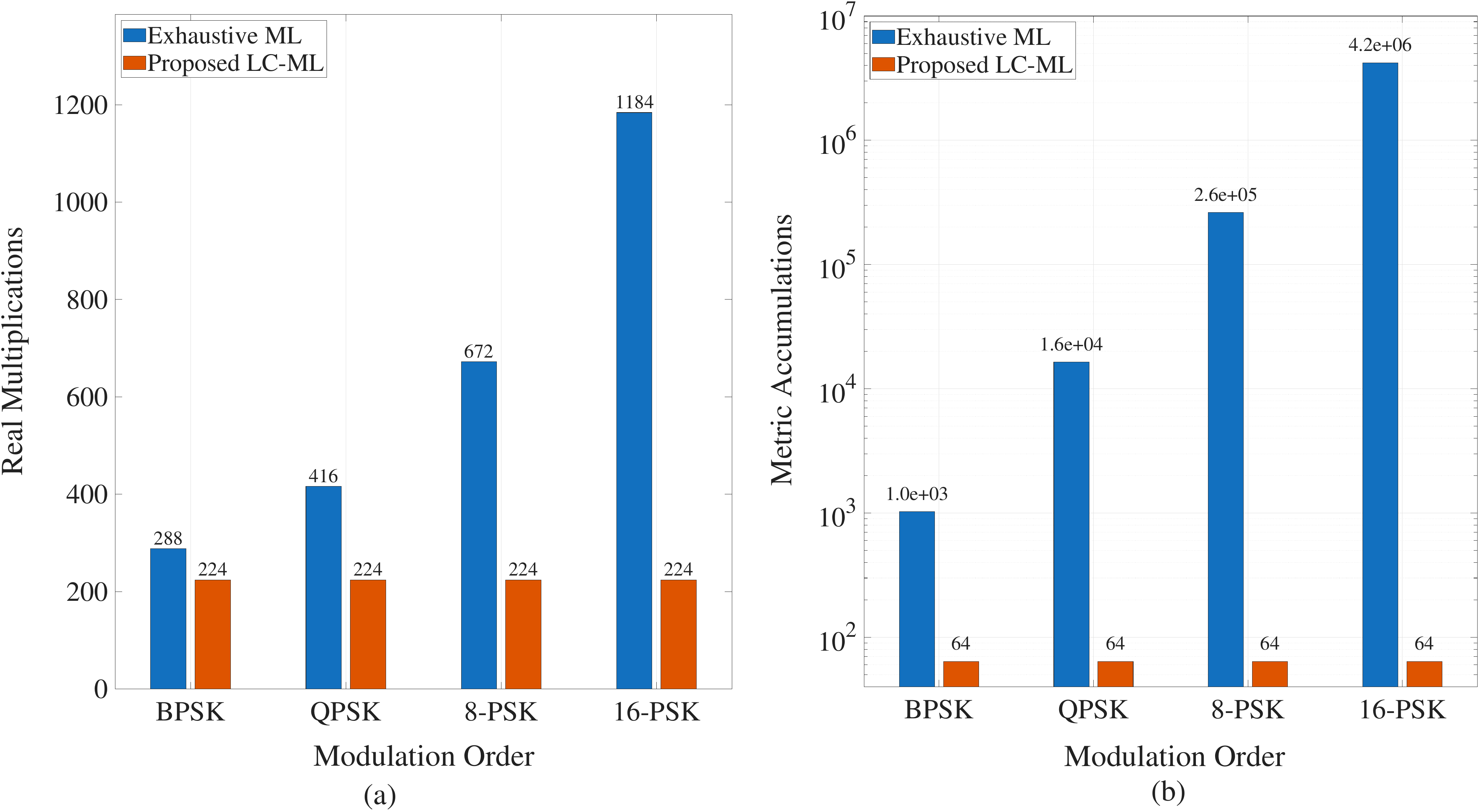}
    \caption{Complexity comparison of exhaustive and proposed LC-ML detection for different PSK orders: (a) real multiplications for column-pair metric preparation; (b) metric accumulations for joint-hypothesis evaluation, with $N_t=4$, $N_r=2$, and $Q=16$.}
    \label{fig:complexity_bar}
\end{figure}

To assess the effect of the active PA-set size, the curves for $N_t=3$ and $N_t=4$ are compared at each value of $N_r$. The two configurations do not operate at the same spectral efficiency. With $M=8$, \eqref{eq:spectral_eff} gives $\eta=3.67$ bps/Hz for $N_t=3$ and $\eta=4$ bps/Hz for $N_t=4$, and the two cannot be matched exactly, since $\lfloor\log_2(N_t!)\rfloor/N_t$ contributes $2/3$ for $N_t=3$ and $1$ for $N_t=4$ while $\log_2(M)$ is an integer. The comparison is therefore made at equal $M$, and $N_t=4$ conveys $9\%$ more information per channel use. At $N_r=1$, this higher rate costs approximately $1$ dB, whereas for $N_r\geq2$ the two configurations perform within $0.1$ dB of each other. Since $r_{\min}=1$ for both codebooks, increasing $N_t$ does not change the minimum codebook-level polynomial SNR exponent. The residual difference reflects the larger number of permutation hypotheses and pairwise error events, which enter the pairwise coding coefficient in \eqref{eq:coding_gain_final}. Additional receive antennas absorb the cost of the higher rate, and increasing the number of active PAs does not degrade PA-DSM performance once $N_r\geq2$.

The analytical and Monte Carlo curves differ in absolute BER because the analysis additionally invokes the high-SNR equivalent-noise approximation, whereas the Monte Carlo results evaluate the complete transmission and detection procedure over randomly generated information matrices and channel realizations. The difference is largest outside the high-SNR range in which that approximation is accurate. Nevertheless, the two sets of curves exhibit consistent slope variations with $N_r$, supporting the receive diversity trend predicted by the asymptotic analysis.\looseness=-1

\vspace{-0.5em}
\subsection{Complexity Analysis}
\vspace{-0.3em}
We compare the arithmetic requirements of the proposed detector with those of an exhaustive implementation. The comparison separately considers real multiplications (RMs) used to form the column-pair metrics and the metric accumulations required to evaluate the joint hypotheses.

The exhaustive detector searches over $N_{\rm hyp}=QM^{N_t}$ valid information matrices. Rather than recomputing the differential metric for every hypothesis, a practical implementation can pre-compute the column-pair symbol costs and reuse them across the joint search. Nevertheless, each of the $QM^{N_t}$ hypotheses requires the accumulation of $N_t$ costs, resulting in
\setlength{\abovedisplayskip}{6pt}
\setlength{\belowdisplayskip}{6pt}
\begin{equation} 
    C_{\rm ex}^{\rm acc} \approx QN_tM^{N_t} 
\label{eq:exhaustive_accumulations} 
\end{equation} 
metric accumulations, followed by the corresponding comparisons.
For the proposed detector, Phase~1 computes $N_t^2$ length-$N_r$ inner products, the norms of the $2N_t$ received columns, and the costs associated with the phase-quantized PSK symbols. Using four RMs per complex multiplication and two RMs per squared magnitude, the corresponding RM count is approximately \looseness=-1
\setlength{\abovedisplayskip}{4pt}
\setlength{\belowdisplayskip}{4pt}
\begin{equation} 
\begin{aligned} 
    C_{\rm LC\text{-}ML}^{\rm RM} &\approx 4N_rN_t^2 +4N_rN_t +4N_t^2, 
\end{aligned} 
\label{eq:lc_ml_rm_count} 
\end{equation}
excluding the implementation-dependent cost of PSK phase quantization. Phase~2 evaluates the $Q$ admissible permutations using \vspace{-0.5em}
\setlength{\abovedisplayskip}{6pt}
\setlength{\belowdisplayskip}{6pt}
\begin{equation} 
    C_{\rm LC\text{-}ML}^{\rm acc} \approx QN_t 
\label{eq:lc_ml_accumulations} 
\end{equation}
metric accumulations. Hence, with constant-time PSK quantization, the overall complexity order is $\mathcal O\left(N_rN_t^2+QN_t\right)$.
%
For $N_t=4$, $N_r=2$, $M=8$, and $Q=16$, \eqref{eq:lc_ml_rm_count} gives approximately $224$ RMs, while Phase~2 requires $64$ metric accumulations, against $QN_tM^{N_t}=262\,144$ for the exhaustive search. Fig.~\ref{fig:complexity_bar} compares these counts across modulation orders. The proposed detector therefore removes the $M^{N_t}$ joint
modulation-combination search while returning the same minimizer of \eqref{eq:ml_detector}.

\vspace{-0.5em}
\section{Conclusion}

\label{sec:conclusion}

This paper has proposed PA-DSM, a non-coherent system combining the spatial reconfigurability of waveguide-fed PAs with differential encoding. RSSI-assisted selection adapts the active PA subset to the UE location and large-scale propagation conditions without estimating the complete instantaneous channel matrix. The PSK-specialized LC-ML detector removes the $M^{N_t}$ joint modulation-combination search while preserving the minimizer of the exhaustive differential metric. Under normalized common-$K$ Rician fading, the MGF-based analysis separated the scattered and deterministic LoS contributions to the pairwise error probability. It has shown that the previous differential state preserves the pairwise rank and eigenvalues but can change the LoS contribution through spatial alignment. The analysis has further identified the rank-dependent polynomial SNR exponent and the minimum-rank behavior of the PA-DSM codebook. Results across representative FR3 environments have demonstrated effective spatial adaptation, favorable performance relative to coherent benchmarks when their channel estimates age between pilot transmissions, and substantial complexity reduction.

\appendices
\vspace{-1em}
\section{Proof of Lemma~\ref{lem:mgf}}
\label{app:lemma_mgf}

Using \eqref{eq:effective_distance} and \eqref{eq:effective_difference_matrix}, the channel-dependent squared distance between the transmitted and competing hypotheses can be expressed as
\setlength{\abovedisplayskip}{4pt}
\setlength{\belowdisplayskip}{4pt}
\begin{equation} 
    \xi = \left\| \mathbf{H}\mathbf{\Delta}_{\rm eff} \right\|_F^2 = \operatorname{Tr}\left( \mathbf{H} \mathbf{\Delta}_{\rm eff} \mathbf{\Delta}_{\rm eff}^H \mathbf{H}^H \right).
    \label{eq:app_xi_trace} 
\end{equation}
Substituting the eigendecomposition in (\ref{eq:effective_difference_evd}) into (\ref{eq:app_xi_trace}) gives
\setlength{\abovedisplayskip}{4pt}
\setlength{\belowdisplayskip}{4pt}
\begin{equation}
    \xi =
    \operatorname{Tr}\left(
    \mathbf{H}
    \mathbf{V}\mathbf{\Lambda}\mathbf{V}^H
    \mathbf{H}^H
    \right) =
    \sum_{i=1}^{r}
    \lambda_i
    \left\|
    \mathbf{H}\mathbf{v}_i
    \right\|^2.
    \label{eq:app_xi_eigenmodes}
\end{equation}
For each nonzero eigenmode, define $\mathbf{g}_i\triangleq\mathbf{H}\mathbf{v}_i =\bar{\mathbf{H}}\mathbf{v}_i+ \widetilde{\mathbf{H}}\mathbf{v}_i$.
Since $\mathbf{v}_i$ has unit norm and the entries of $\widetilde{\mathbf{H}}$ are i.i.d. according to $\mathcal{CN}(0,\sigma_h^2)$, the scattered NLoS component projected onto $\mathbf{v}_i$ satisfies \looseness=-1
\setlength{\abovedisplayskip}{6pt}
\setlength{\belowdisplayskip}{6pt}
\begin{equation}
    \widetilde{\mathbf{H}}\mathbf{v}_i \sim \mathcal{CN}\left( \mathbf{0},\sigma_h^2\mathbf{I}_{N_r}\right). \vspace{-0.5em}
    \label{eq:app_projected_scattered_distribution}
\end{equation}
Thus, conditioned on the previous differential state $\mathbf{S}_{k-1}$, 
\setlength{\abovedisplayskip}{4pt}
\setlength{\belowdisplayskip}{4pt}
\begin{equation}
    \mathbf{g}_i \sim \mathcal{CN}\left(\bar{\mathbf{H}}\mathbf{v}_i,\sigma_h^2\mathbf{I}_{N_r}\right).
    \label{eq:app_projected_channel_distribution}
\end{equation}
The projected vectors associated with distinct eigenmodes are mutually independent. Specifically, for $i\neq j$, their cross-covariance is $\mathbb{E}\left[\left(\widetilde{\mathbf{H}}\mathbf{v}_i\right)\left(\widetilde{\mathbf{H}}\mathbf{v}_j \right)^H\right] = \sigma_h^2\left(\mathbf{v}_j^H\mathbf{v}_i\right)\mathbf{I}_{N_r}=\mathbf{0},$
where the last equality follows from the orthogonality of $\mathbf v_i$ and $\mathbf v_j$. Since the projected vectors are jointly complex Gaussian, zero cross-covariance implies independence. Therefore, $\mathbf g_1,\ldots,\mathbf g_r$ are mutually independent conditioned on $\mathbf S_{k-1}$.

Next, we define $\zeta_i \triangleq \frac{2}{\sigma_h^2} \left\| \mathbf{g}_i\right\|^2$. From \eqref{eq:app_projected_channel_distribution}, $\zeta_i$ follows a real noncentral chi-square distribution with $2N_r$ degrees of freedom, and the noncentrality parameter is determined by the deterministic LoS component and is given by $\delta_i = \frac{2}{\sigma_h^2} \left\|\bar{\mathbf{H}}\mathbf{v}_i \right\|^2 =\frac{ 2\mu_i(\mathbf{S}_{k-1})}{\sigma_h^2}$,
where $\mu_i(\mathbf{S}_{k-1}) \triangleq \left\| \bar{\mathbf{H}}\mathbf{v}_i \right\|^2$, as defined in \eqref{eq:noncentrality_effective}. Hence, $\|\mathbf{g}_i\|^2$ is a scaled noncentral chi-square random variable.
Using the MGF of the noncentral chi-square variable $\zeta_i$, the $i$th eigenmode contribution $\lambda_i\|\mathbf{g}_i\|^2$ has the MGF \looseness=-1
\setlength{\abovedisplayskip}{4pt}
\setlength{\belowdisplayskip}{4pt}
\begin{equation}
\begin{aligned}
    M_i(\nu)&\triangleq \mathbb{E}\left[\exp\left(\nu\lambda_i\|\mathbf{g}_i\|^2\right)\mid\mathbf{S}_{k-1}
    \right]
    \\
    &=\left(1-\nu\lambda_i\sigma_h^2\right)^{-N_r}\exp\left(\frac{\nu\lambda_i\mu_i(\mathbf{S}_{k-1})}{
    1-\nu\lambda_i\sigma_h^2}\right).
\end{aligned}
\label{eq:app_eigenmode_mgf}
\end{equation}
The MGF exists for $\nu<1/(\lambda_i\sigma_h^2)$, which follows from the convergence condition of the noncentral chi-square MGF.

Finally, since the eigenmode contributions in \eqref{eq:app_xi_eigenmodes} are mutually independent conditioned on $\mathbf{S}_{k-1}$, the conditional MGF of $\xi$ is the product of their individual MGFs 
\setlength{\abovedisplayskip}{4pt}
\setlength{\belowdisplayskip}{4pt}
\begin{equation}
    M_{\xi\mid\mathbf{S}_{k-1}}(\nu)=\prod_{i=1}^{r}M_i(\nu).
    \label{eq:app_mgf_product}
\end{equation}
This expression is valid for $\nu<1/({\lambda_{\max}\sigma_h^2})$, where $\lambda_{\max}=\mathop{\mathrm{max}}\limits_{1\leq i\leq r}\lambda_i$ ensures the convergence of all eigenmode MGFs.
Equations \eqref{eq:app_eigenmode_mgf} and
\eqref{eq:app_mgf_product} yield \eqref{eq:mgf_exact}, completing the proof. 

\vspace{-0.5em}
\section{Proof of Lemma \ref{lem:perm_invariance}} 
\label{app:state_reduction}

Each information matrix can be written as $\mathbf X_k=\mathbf A_{q_k}\mathbf D_k$, where $\mathbf A_{q_k}$ is a permutation matrix and $\mathbf D_k$ is a diagonal PSK matrix. Since the product of monomial matrices remains monomial, the accumulated differential state admits the decomposition $\mathbf S_{k-1}=\mathbf P_{k-1}\mathbf\Phi_{k-1}$, where $\mathbf P_{k-1}$ is the accumulated permutation matrix and $\mathbf\Phi_{k-1}$ is a diagonal matrix containing the accumulated PSK phases. From Lemma~\ref{lem:mgf}, the state dependence of the conditional PEP is introduced only through the LoS-related terms
\setlength{\abovedisplayskip}{4pt}
\setlength{\belowdisplayskip}{4pt}
\begin{equation}
    \mu_i(\mathbf S_{k-1}) = \lvert| \bar{\mathbf H}\mathbf v_i \rvert|^2,
\end{equation}
where $\mathbf v_i$ denotes the eigenvectors of
$\mathbf S_{k-1}\mathbf\Delta\mathbf\Delta^H\mathbf S_{k-1}^H$.
Since $\mathbf S_{k-1}$ is unitary,
\setlength{\abovedisplayskip}{4pt}
\setlength{\belowdisplayskip}{4pt}
\begin{equation}
    \mathbf S_{k-1}\mathbf\Delta\mathbf\Delta^H \mathbf S_{k-1}^H = \mathbf P_{k-1} \mathbf\Phi_{k-1} \mathbf\Delta\mathbf\Delta^H \mathbf\Phi_{k-1}^H \mathbf P_{k-1}^H.
    \label{eq:71}
\end{equation}

Therefore, $\mathbf v_i =\mathbf P_{k-1}\mathbf\Phi_{k-1}\mathbf u_i$, where $\mathbf u_i$ is the corresponding eigenvector of $\mathbf\Delta\mathbf\Delta^H$. 
Under the homogeneous LoS model in \eqref{eq:homogeneous_rician_mean}, $\bar{\mathbf H}\mathbf P_{k-1}=\bar{\mathbf H}$ for any permutation matrix $\mathbf P_{k-1}$.
Consequently,\looseness=-1
\setlength{\abovedisplayskip}{4pt}
\setlength{\belowdisplayskip}{4pt}
\begin{equation}
    \mu_i(\mathbf S_{k-1})=\left|\bar{\mathbf H}\mathbf P_{k-1}\mathbf\Phi_{k-1}\mathbf u_i\right|^2=\left|\bar{\mathbf H}\mathbf\Phi_{k-1}\mathbf u_i\right|^2,
    \label{eq:LoS_cont_term}
\end{equation}
which is independent of $\mathbf P_{k-1}$.

Since the remaining terms in the conditional PEP expression in \eqref{eq:explicit_channel_averaged_pep} depend only on the pairwise difference eigenvalues and channel parameters, the PEP is invariant to the permutation component of the previous differential state.

\vspace{-0.6cm}
\section{Proof of Theorem \ref{thm:asymptotic}}
\label{app:asymptotic_pep}
\vspace{-0.5em}
Consider a fixed representative phase state $\mathbf{\Phi}_{k-1}\in\mathcal F$. By Lemma~\ref{lem:perm_invariance}, the PEP is independent of the permutation component of the previous differential state. Hence, \eqref{eq:mgf_pep} can be evaluated by setting $\mathbf{S}_{k-1}=\boldsymbol{\Phi}_{k-1}$, for which the eigenvectors in Lemma~\ref{lem:mgf} become $\mathbf{v}_i=\mathbf{\Phi}_{k-1}\mathbf{u}_i$ with $\mathbf{\Delta}\mathbf{\Delta}^{H}\mathbf{u}_i=\lambda_i\mathbf{u}_i$, so the eigenvalues are unchanged. For $0<\theta\leq\pi/2$, define \looseness=-1
%
\setlength{\abovedisplayskip}{4pt}
\setlength{\belowdisplayskip}{4pt}
\begin{equation} 
    \nu_{\theta} \triangleq -\frac{\rho}{8\sin^2\theta} \leq -\frac{\rho}{8}, \label{eq:app_mgf_argument_bound} 
\end{equation} 
where the inequality follows from $\sin^2\theta\leq1$. Since $\xi\geq0$ by \eqref{eq:effective_distance}, the conditional MGF is nondecreasing in its real argument:\looseness=-1
\setlength{\abovedisplayskip}{4pt}
\setlength{\belowdisplayskip}{4pt}
\begin{equation} 
\begin{aligned} 
    \frac{\mathrm d}{\mathrm d\nu} M_{\xi\mid\mathbf\Phi_{k-1}}(\nu) &= \mathbb E_{\mathbf H} \left[ \xi\exp(\nu\xi) \mid\mathbf\Phi_{k-1} \right]\geq0. 
\end{aligned} 
\label{eq:app_mgf_monotonicity} 
\end{equation}

It follows from \eqref{eq:app_mgf_argument_bound} and \eqref{eq:app_mgf_monotonicity} that
\setlength{\abovedisplayskip}{4pt}
\setlength{\belowdisplayskip}{4pt}
\begin{equation}
   M_{\xi\mid\mathbf\Phi_{k-1}}(\nu_{\theta}) \leq M_{\xi\mid\mathbf\Phi_{k-1}} \left( -\frac{\rho}{8} \right).
   \label{eq:app_mgf_pointwise_bound}
\end{equation}

Applying this pointwise bound to the integral in \eqref{eq:mgf_pep} gives
\setlength{\abovedisplayskip}{4pt}
\setlength{\belowdisplayskip}{4pt}
\begin{equation}
\begin{aligned}
    &P_{\mathrm a} \left( \mathbf X\rightarrow\hat{\mathbf X} \mid\mathbf\Phi_{k-1} \right)\leq \frac{1}{2} M_{\xi\mid\mathbf\Phi_{k-1}} \left( -\frac{\rho}{8} \right).
\end{aligned}
\label{eq:app_chernoff_bound}
\end{equation}

We next evaluate the MGF on the right-hand side of \eqref{eq:app_chernoff_bound} as $\rho\rightarrow\infty$. Substituting $\nu=-\rho/8$ into the scattered-channel factor in \eqref{eq:mgf_exact} gives

\setlength{\abovedisplayskip}{4pt}
\setlength{\belowdisplayskip}{4pt}
\begin{equation}
\begin{aligned}
    1-\nu\lambda_i\sigma_h^2 &= 1+\frac{\rho}{8}\lambda_i\sigma_h^2= \frac{\rho}{8}\lambda_i\sigma_h^2 \left[ 1+o(1) \right].
\end{aligned}
\label{eq:app_polynomial_factor}
\end{equation}
Therefore, 
\setlength{\abovedisplayskip}{4pt}
\setlength{\belowdisplayskip}{4pt}
\begin{equation}
    \prod_{i=1}^{r} \left( 1+\frac{\rho}{8}\lambda_i\sigma_h^2 \right)^{-N_r}= \left[ \frac{\rho}{8} \left( \prod_{i=1}^{r}\lambda_i \right)^{1/r} \sigma_h^2 \right]^{-rN_r} \left[ 1+o(1) \right].
\label{eq:app_polynomial_product}
\end{equation}

For the deterministic LoS factor in \eqref{eq:mgf_exact}, substituting $\nu=-\rho/8$ gives \looseness=-1

\setlength{\abovedisplayskip}{7pt}
\setlength{\belowdisplayskip}{7pt}
\begin{equation} 
\begin{aligned} 
    &\exp \left( -\frac{ (\rho/8)\lambda_i \mu_i(\mathbf\Phi_{k-1}) }{ 1+(\rho/8)\lambda_i\sigma_h^2 } \right)= \exp \left( -\frac{ \mu_i(\mathbf\Phi_{k-1}) }{ \sigma_h^2+ \dfrac{8}{\rho\lambda_i} } \right)\\ &\quad= \exp \left( -\frac{ \mu_i(\mathbf\Phi_{k-1}) }{ \sigma_h^2 } \right) \left[ 1+o(1) \right], \qquad \rho\rightarrow\infty. 
\end{aligned} 
\label{eq:app_exponential_factor} 
\end{equation}
Substituting \eqref{eq:app_polynomial_product} and \eqref{eq:app_exponential_factor} into \eqref{eq:mgf_exact} yields

\setlength{\abovedisplayskip}{6pt}
\setlength{\belowdisplayskip}{6pt}
\begin{equation}
\begin{aligned}
    &M_{\xi\mid\mathbf\Phi_{k-1}} \left( -\frac{\rho}{8} \right)= \left[ \frac{\rho}{8} \left( \prod_{i=1}^{r}\lambda_i \right)^{1/r} \sigma_h^2 \right]^{-rN_r}\\ &\quad\times \exp \left( -\sum_{i=1}^{r} \frac{ \mu_i(\mathbf\Phi_{k-1}) }{ \sigma_h^2 } \right) \left[ 1+o(1) \right].
\end{aligned}
\label{eq:app_mgf_asymptote}
\end{equation}

Finally, substituting \eqref{eq:app_mgf_asymptote} into \eqref{eq:app_chernoff_bound} and using the definitions of $G_d$, $G_c$, and $\Gamma_{\rm PA}(\mathbf\Phi_{k-1})$ in \eqref{eq:pairwise_diversity}, \eqref{eq:coding_gain_final}, and \eqref{eq:pa_decay_general}, respectively, gives \looseness=-1 
\setlength{\abovedisplayskip}{4pt}
\setlength{\belowdisplayskip}{4pt}
\begin{equation}
\begin{aligned}
    &P_{\mathrm a} \left( \mathbf X\rightarrow\hat{\mathbf X} \mid\mathbf\Phi_{k-1} \right)\\ &\quad\leq \frac{1}{2} \left( G_c\rho \right)^{-G_d} \exp \left[ -\Gamma_{\rm PA}(\mathbf\Phi_{k-1}) \right] \left[ 1+o(1) \right],
\end{aligned}
\end{equation}
which proves Theorem~\ref{thm:asymptotic}.

\vspace{-0.5em}

\bibliographystyle{IEEEtran}
\vspace{-0.5em}
\bibliography{references}

\begin{thebibliography}{10}
\providecommand{\url}[1]{#1}
\csname url@samestyle\endcsname
\providecommand{\newblock}{\relax}
\providecommand{\bibinfo}[2]{#2}
\providecommand{\BIBentrySTDinterwordspacing}{\spaceskip=0pt\relax}
\providecommand{\BIBentryALTinterwordstretchfactor}{4}
\providecommand{\BIBentryALTinterwordspacing}{\spaceskip=\fontdimen2\font plus
\BIBentryALTinterwordstretchfactor\fontdimen3\font minus
  \fontdimen4\font\relax}
\providecommand{\BIBforeignlanguage}[2]{{%
\expandafter\ifx\csname l@#1\endcsname\relax
\typeout{** WARNING: IEEEtran.bst: No hyphenation pattern has been}%
\typeout{** loaded for the language `#1'. Using the pattern for}%
\typeout{** the default language instead.}%
\else
\language=\csname l@#1\endcsname
\fi
#2}}
\providecommand{\BIBdecl}{\relax}
\BIBdecl

\bibitem{11594671}
Y.~Akar \emph{et~al.}, ``Pinching antenna-assisted differential spatial
  modulation,'' in \emph{Proc. 32nd Int. Conf. Telecommun. (ICT)},
  Thessaloniki, Greece, May 20--22, 2026.

\bibitem{11456641}
W.~Chen \emph{et~al.}, ``Toward standardization of {6G} and {NextG}: Key
  technologies to enable fundamental enhancements,'' \emph{IEEE J. Sel. Areas
  Commun.}, vol.~44, pp. 4333--4365, 2026.

\bibitem{7448967}
E.~Basar, ``On multiple-input multiple-output {OFDM} with index modulation for
  next generation wireless networks,'' \emph{IEEE Trans. Signal Process.},
  vol.~64, no.~15, pp. 3868--3878, 2016.

\bibitem{10659349}
A.~T. Dogukan, E.~Arslan, and E.~Basar, ``Reconfigurable intelligent
  surface-enabled downlink {NOMA},'' \emph{IEEE Trans. Wireless Commun.},
  vol.~23, no.~11, pp. 16\,950--16\,961, 2024.

\bibitem{10373568}
E.~Basar, ``Noise modulation,'' \emph{IEEE Wireless Commun. Lett.}, vol.~13,
  no.~3, pp. 844--848, 2024.

\bibitem{11032161}
E.~Yapici \emph{et~al.}, ``Noise modulation over wireless energy transfer:
  {JEIH-NoiseMod},'' \emph{IEEE Wireless Commun. Lett.}, vol.~14, no.~9, pp.
  2768--2772, 2025.

\bibitem{8796365}
E.~Basar \emph{et~al.}, ``Wireless communications through reconfigurable
  intelligent surfaces,'' \emph{IEEE Access}, vol.~7, pp. 116\,753--116\,773,
  2019.

\bibitem{10318061}
L.~Zhu, W.~Ma, and R.~Zhang, ``Modeling and performance analysis for movable
  antenna enabled wireless communications,'' \emph{IEEE Trans. Wireless
  Commun.}, vol.~23, no.~6, pp. 6234--6250, 2024.

\bibitem{9264694}
K.-K. Wong, A.~Shojaeifard, K.-F. Tong, and Y.~Zhang, ``Fluid antenna
  systems,'' \emph{IEEE Trans. Wireless Commun.}, vol.~20, no.~3, pp.
  1950--1962, 2021.

\bibitem{NTTDOCOMO}
A.~Fukuda \emph{et~al.}, ``Pinching antenna---{U}sing a dielectric waveguide as
  an antenna,'' \emph{NTT DOCOMO Tech. J.}, vol.~23, no.~3, pp. 5--12, Jan.
  2022.

\bibitem{10605910}
D.~Shakya \emph{et~al.}, ``Comprehensive {FR1(C)} and {FR3} lower and upper
  mid-band propagation and material penetration loss measurements and channel
  models in indoor environment for {5G} and {6G},'' \emph{IEEE Open J. Commun.
  Soc.}, vol.~5, pp. 5192--5218, Jul. 2024.

\bibitem{11327450}
H.~Xu \emph{et~al.}, ``Near-field propagation and spatial non-stationarity
  channel model for 6–24~{GHz} ({FR3}) extremely large-scale {MIMO}: Adopted
  by {3GPP} for {6G},'' \emph{IEEE J. Sel. Areas Commun.}, vol.~44, pp.
  3201--3218, 2026.

\bibitem{11172334}
M.~Samy \emph{et~al.}, ``Pinching antenna systems versus reconfigurable
  intelligent surfaces in {mmWave},'' \emph{IEEE Wireless Commun. Lett.},
  vol.~14, no.~12, pp. 4022--4026, 2025.

\bibitem{11368709}
S.~Yang \emph{et~al.}, ``Pinching-antenna-assisted index modulation: Channel
  modeling, transceiver design, and performance analysis,'' \emph{IEEE Trans.
  Wireless Commun.}, vol.~25, pp. 11\,497--11\,512, 2026.

\bibitem{he2025risassisteddownlinkpinchingantennasystems}
C.~He \emph{et~al.}, ``{RIS}-assisted downlink pinching-antenna systems:
  {GNN}-enabled optimization approaches,'' arXiv:2511.20305, 2025.

\bibitem{11204499}
B.~Zhuo \emph{et~al.}, ``{P-NOMA} for pinching-antenna systems {(PASS)},''
  \emph{IEEE Wireless Commun. Lett.}, vol.~15, pp. 131--135, 2026.

\bibitem{11414143}
A.~Bereyhi \emph{et~al.}, ``{MIMO-PASS}: Uplink and downlink transmission via
  {MIMO} pinching-antenna systems,'' \emph{IEEE Trans. Commun.}, vol.~74, pp.
  5701--5716, 2026.

\bibitem{11303890}
E.~Illi, M.~Qaraqe, and A.~Ghrayeb, ``Secure pinching antenna-aided {ISAC},''
  \emph{IEEE Commun. Lett.}, vol.~30, pp. 727--731, 2026.

\bibitem{11205176}
Y.~Zhong \emph{et~al.}, ``Physical layer security for pinching-antenna systems
  via index and directional modulation,'' \emph{IEEE Wireless Commun. Lett.},
  vol.~15, pp. 230--234, 2026.

\bibitem{11314615}
M.~Liu \emph{et~al.}, ``Integrated sensing and communication with index
  modulation over pinching antennas,'' \emph{IEEE Commun. Lett.}, vol.~30, pp.
  737--741, 2026.

\bibitem{8004416}
E.~Basar \emph{et~al.}, ``Index modulation techniques for next-generation
  wireless networks,'' \emph{IEEE Access}, vol.~5, pp. 16\,693--16\,746, Sep.
  2017.

\bibitem{11160744}
M.~Ying \emph{et~al.}, ``Upper mid-band channel measurements and
  characterization at 6.75 {GHz} {FR1}({C}) and 16.95 {GHz} {FR3} in an indoor
  factory scenario,'' in \emph{Proc. IEEE Int. Conf. Commun. (ICC)}, Montreal,
  QC, Canada, Jun. 8--12, 2025, pp. 3303--3308.

\bibitem{11161884}
D.~Shakya \emph{et~al.}, ``Urban outdoor propagation measurements and channel
  models at {6.75 GHz FR1(C)} and {16.95 GHz FR3} upper mid-band spectrum for
  {5G} and {6G},'' in \emph{Proc. {IEEE} Int. Conf. Commun. ({ICC})}, Montreal,
  QC, Canada, Jun. 8--12, 2025, pp. 3291--3296.

\bibitem{11184847}
P.~Zhang \emph{et~al.}, ``Fluid antenna-assisted rectangular differential index
  modulation: A non-coherent system design, optimization, and performance
  analysis,'' \emph{IEEE J. Sel. Areas Commun.}, vol.~44, pp. 1307--1321, 2026.

\bibitem{Tao_DSM_PA}
Y.~Tao \emph{et~al.}, ``Differential spatial modulation with transmit diversity
  for pinching-antenna systems,'' arXiv:2605.04578, 2026.

\bibitem{10855589}
R.-Y. Wei and C.-Y. Chen, ``Low-complexity maximum-likelihood detectors
  incorporating pre-calculated symbol metrics for differential spatial
  modulation,'' \emph{IEEE Wireless Commun. Lett.}, vol.~14, no.~4, pp.
  1124--1128, Apr. 2025.

\bibitem{11202577}
Z.~Wang \emph{et~al.}, ``Modeling and beamforming optimization for
  pinching-antenna systems,'' \emph{IEEE Trans. Commun.}, vol.~73, no.~12, pp.
  13\,904--13\,919, Dec. 2025.

\bibitem{11657464}
K.~Wang, Z.~Ding, and L.~Hanzo, ``Generalized pinching-antenna systems: A
  leaky-coaxial-cable perspective,'' \emph{IEEE Trans. Commun.}, vol.~74, pp.
  12\,774--12\,787, 2026.

\bibitem{Zhang2025PASS_RSSI}
Y.~Zhang \emph{et~al.}, ``Pinching-antenna systems {(PASS)}-based indoor
  positioning,'' arXiv:2508.08185, 2025.

\bibitem{6879496}
Y.~Bian \emph{et~al.}, ``Differential spatial modulation,'' \emph{IEEE Trans.
  Veh. Technol.}, vol.~64, no.~7, pp. 3262--3268, Jul. 2015.

\bibitem{258319}
J.~W. Craig, ``A new, simple and exact result for calculating the probability
  of error for two-dimensional signal constellations,'' in \emph{Proc. IEEE
  Mil. Commun. Conf. (MILCOM)}, vol.~2, McLean, VA, USA, Nov. 04--07, 1991, pp.
  571--575.

\bibitem{simon2005digital}
M.~K. Simon and M.-S. Alouini, \emph{Digital Communication over Fading
  Channels}, 2nd~ed.\hskip 1em plus 0.5em minus 0.4em\relax Hoboken, NJ, USA:
  Wiley, 2005.

\bibitem{10901735}
D.~Shakya \emph{et~al.}, ``Propagation measurements and channel models in
  indoor environment at 6.75 {GHz} {FR1}({C}) and 16.95 {GHz} {FR3} upper-mid
  band spectrum for {5G} and {6G},'' in \emph{Proc. IEEE Global Commun. Conf.
  (GLOBECOM)}, Cape Town, South Africa, Dec. 8--12, 2024, pp. 998--1003.

\bibitem{8620255}
{\"O}.~{\"O}zdogan, E.~Bj{\"o}rnson, and E.~G. Larsson, ``Massive {MIMO} with
  spatially correlated {Rician} fading channels,'' \emph{IEEE Trans. Commun.},
  vol.~67, no.~5, pp. 3234--3250, May 2019.

\bibitem{10945421}
Z.~Ding, R.~Schober, and H.~V. Poor, ``Flexible-antenna systems: A
  pinching-antenna perspective,'' \emph{IEEE Trans. Commun.}, vol.~73, no.~10,
  pp. 9236--9253, Oct. 2025.

\bibitem{8645135}
{\"O}.~{\"O}zdogan, E.~Bj{\"o}rnson, and J.~Zhang, ``Cell-free massive {MIMO}
  with {Rician} fading: Estimation schemes and spectral efficiency,'' in
  \emph{Proc. 52nd Asilomar Conf. Signals, Syst., Comput.}, Pacific Grove, CA,
  USA, Oct. 28--31, 2018, pp. 975--979.

\bibitem{10685086}
B.~A. Ozden and E.~Aydin, ``Antenna selection for receive spatial modulation
  system empowered by reconfigurable intelligent surface,'' \emph{IEEE Trans.
  Veh. Technol.}, vol.~74, no.~1, pp. 1169--1179, 2025.

\bibitem{Clarke}
R.~H. Clarke, ``A statistical theory of mobile-radio reception,'' \emph{Bell
  Syst. Tech. J.}, vol.~47, no.~6, pp. 957--1000, 1968.

\bibitem{5672371}
E.~Başar, U.~Aygölü, E.~Panayirci, and H.~V. Poor, ``Space-time block coded
  spatial modulation,'' \emph{IEEE Trans. Commun.}, vol.~59, no.~3, pp.
  823--832, 2011.

\bibitem{965648}
M.~Tao and R.~S. Cheng, ``Differential space-time block codes,'' in \emph{Proc.
  IEEE Global Commun. Conf. (GLOBECOM)}, vol.~2, San Antonio, TX, USA, Nov.
  2001, pp. 1098--1102.

\end{thebibliography}

\end{document}